\documentclass[journal]{IEEEtran}
\usepackage{amsmath}
\usepackage{amsthm}
\usepackage{amssymb}
\usepackage{graphicx}
\usepackage{subfigure}
\usepackage{color}
\usepackage{balance}
\usepackage{cite}
\usepackage{hyperref}
\newtheorem{definition}{Definition}
\newtheorem{lemma}{Lemma}
\newtheorem{corollary}{Corollary}
\newtheorem{theorem}{Theorem}
\begin{document}
\title{Non-Asymptotic Delay Bounds for Multi-Server Systems with Synchronization Constraints}
\author{\IEEEauthorblockN{Markus Fidler\thanks{This work was supported in part by the European Research Council (ERC) under Starting Grant UnIQue (StG 306644). This manuscript is a revised and extended version of the paper~\cite{fidler:forkjoin} that appeared in the IEEE Infocom 2016 proceedings.}, Brenton Walker, and Yuming Jiang}\thanks{M. Fidler and B. Walker are with the Institute of Communications Technology, Leibniz Universit\"at Hannover, Germany, (E-mail: markus.fidler@ikt.uni-hannover.de and brenton.walker@ikt.uni-hannover.de). Y. Jiang is with the Department of Telematics, NTNU Trondheim, Norway (E-mail: jiang@item.ntnu.no).}}
\maketitle
\begin{abstract}
Multi-server systems have received increasing attention with important implementations such as Google MapReduce, Hadoop, and Spark. Common to these systems are a fork operation, where jobs are first divided into tasks that are processed in parallel, and a later join operation, where completed tasks wait until the results of all tasks of a job can be combined and the job leaves the system. The synchronization constraint of the join operation makes the analysis of fork-join systems challenging and few explicit results are known. In this work, we model fork-join systems using a max-plus server model that enables us to derive statistical bounds on waiting and sojourn times for general arrival and service time processes. We contribute end-to-end delay bounds for multi-stage fork-join networks that grow in $\mathcal{O}(h \ln k)$ for $h$ fork-join stages, each with $k$ parallel servers. We perform a detailed comparison of different multi-server configurations and highlight their pros and cons. We also include an analysis of single-queue fork-join systems that are non-idling and achieve a fundamental performance gain, and compare these results to both simulation and a live Spark system.
\end{abstract}
%
%
\section{Introduction}
\label{sec:introduction}
Fork-join systems are an essential model of parallel data processing, e.g., Google MapReduce~\cite{dean:mapreduce}, Hadoop, or Spark~\cite{spark-usenix}, where jobs are divided into $k$ tasks (fork) that are processed in parallel by $k$ servers. Once all tasks of a job are completed, the results are combined (join) and the job leaves the system. Fig.~\ref{fig:fjqueue} illustrates an example. Multi-stage fork-join networks comprise several fork-join systems in tandem, where all tasks of a job have to be completed at the current stage before the job is handed over to the next stage. The difficulty in analyzing such systems is due to a) the statistical dependence of the workload on the parallel servers that is due to the common arrival process~\cite{baccelli:forkjoin,kemper:forkjoin}, and b) the synchronization required by the join operation~\cite{baccelli:forkjoin,tan:mapreduce}.

Significant research has been performed to analyze the performance of fork-join systems. However, exact results are known only for few specific systems, such as two parallel M$\mid$M$\mid$1 queues~\cite{flatto:forkjoin, nelson:forkjoin}. For more complex systems, approximation techniques, e.g.,~\cite{nelson:forkjoin, lebrecht:forkjoin, ko:forkjoin, tan:forkjoin, varma:forkjoin, varki:forkjoin, kemper:forkjoin, alomari:forkjoin}, and bounds, using stochastic orderings~\cite{baccelli:forkjoin}, martingales~\cite{rizk:forkjoin}, or stochastic burstiness constraints~\cite{kesidis:forkjoin}, have been explored. Given the difficulties posed by single-stage fork-join systems, few works consider multi-stage networks. A notable exception is~\cite{varki:forkjoin} where an approximation for closed fork-join networks is developed.

Related synchronization problems also occur in the case of load balancing using parallel servers and in the case of multi-path routing of packet data streams~\cite{han:resequencing} using multi-path protocols~\cite{rizk:forkjoin}. The tail behavior of delays in multi-path routing is investigated in~\cite{han:resequencing} as well as in~\cite{xia:resequencing,gao:resequencing} where large deviation results of resequencing delays for parallel M$\mid$M$\mid$1 queues are derived.

\begin{figure}
  \centering
  \includegraphics[width=0.85\columnwidth]{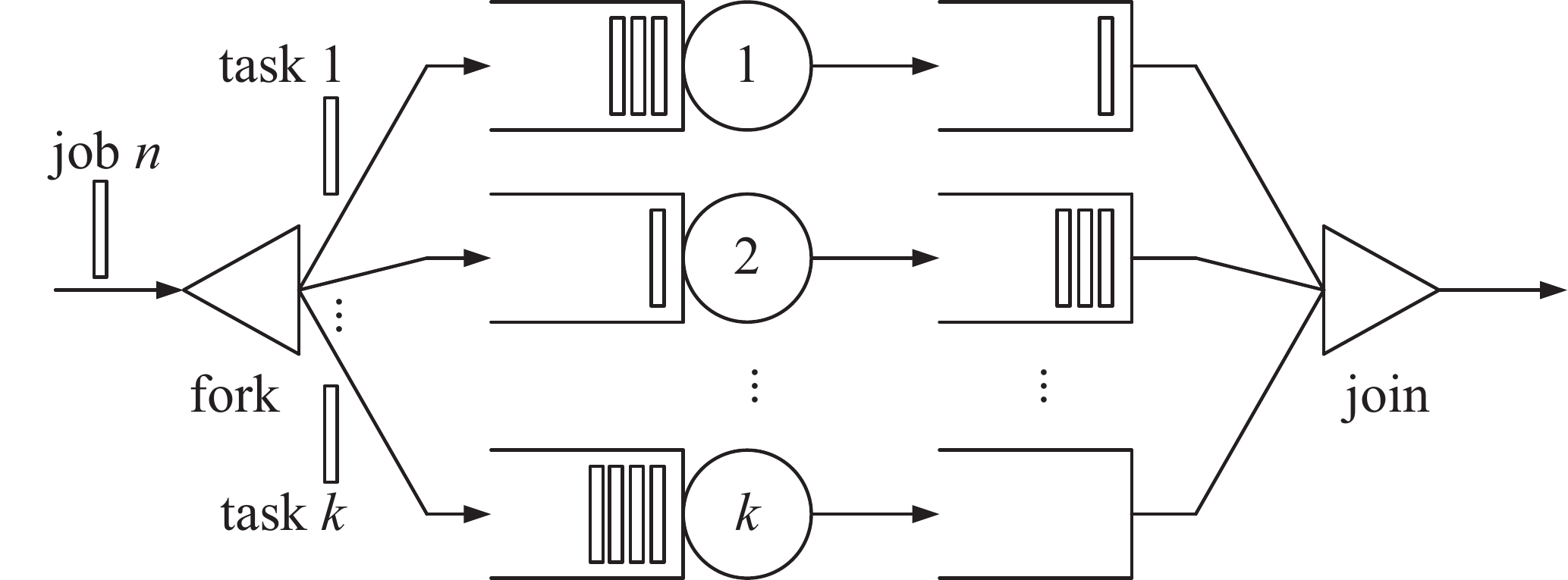}
  \caption{Fork-join system. Each job is composed of $k$ tasks with individual service requirements, that are distributed to $k$ servers (fork). Once all tasks of a job are completed, the job leaves the system (join), i.e., the tasks of a job wait at the join step until all tasks of the job are completed.}
  \label{fig:fjqueue}
\end{figure}

Split-merge systems are a variant of fork-join systems with a stricter synchronization constraint: all tasks of a job have to start execution simultaneously. In contrast, in a fork-join system, the start times of tasks are not synchronized. Split-merge systems are solvable to some extent as they can be expressed as a single server queue where the service process is governed by the service time of the maximal task of each job~\cite{harrison:splitmerge, lebrecht:forkjoin, rizk:forkjoin, joshi:knforkjoin}.

Most closely related to this work are three recent papers~\cite{kesidis:forkjoin, rizk:forkjoin, poloczek:parallelsystems} that employ similar methods. The work~\cite{kesidis:forkjoin} considers single-stage fork-join systems with load balancing, general arrivals of the type defined in~\cite{yin:generalizedstochasticallyboundedburstiness}, and deterministic service. A service curve characterization of fork-join systems is provided and statistical delay bounds are presented. The paper~\cite{rizk:forkjoin} contributes delay bounds for single-stage fork-join systems with renewal as well as Markov-modulated inter-arrival times and independent and identically distributed (iid) service times. The authors prove that delays for fork-join systems grow as $\mathcal{O}(\ln k)$ for $k$ parallel servers, as also found in~\cite{baccelli:forkjoin}. Split-merge systems are shown to have inferior performance~\cite{rizk:forkjoin}, where the stability region of $k$ parallel M$\mid$M$\mid$1 queues decreases with $\ln k$. The work also includes a first application to multi-path routing, assuming a generic window-based protocol that operates on batches of packets. The authors conclude that multi-path routing is only beneficial in the case of two parallel paths and moderate to high utilization. Otherwise, resequencing delays are found to dominate. In~\cite{poloczek:parallelsystems}, the authors evaluate different task assignment policies for parallel server systems with task replication, considering the effects of correlated replicas. Task replication relates to the more general concept of $(k,l)$ fork-join systems~\cite{joshi:knforkjoin, fidler:forkjoin}, where a job is considered completed once $l$ out of $k$ tasks have finished service.

While~\cite{rizk:forkjoin} focuses on split-merge vs. fork-join systems with iid service times, we also consider the case of non-iid service, where we are able to generalize important results, such as the growth of delays in $\mathcal{O}(\ln k)$ for fork-join systems with $k$ parallel servers. Furthermore, we show that fork-join systems can be formulated as a server under the max-plus algebra~\cite{chang:performanceguarantees}. This essential lemma enables the analysis of multi-stage fork-join networks. For $h$ statistically independent fork-join stages each with $k$ parallel servers, we prove that the growth of delays is in $\mathcal{O}(h \ln k)$. The result compares to a scaling in $\mathcal{O}(h \ln (hk))$ that we obtained previously in~\cite{fidler:forkjoin} without assuming independence of the stages.

We perform a detailed evaluation of different multi-server configurations which reveals that fork-join systems mostly but not universally outperform classical multi-server systems. Beyond~\cite{fidler:forkjoin, kesidis:forkjoin, rizk:forkjoin, poloczek:parallelsystems}, we also include single-queue multi-server as well as single-queue fork-join systems. In contrast to the standard fork-join model, where each of the servers has an individual queue, single-queue systems are non-idling in the sense that queueing can occur only if all servers are busy. Our evaluation reveals a fundamental performance gain of non-idling single-queue systems. We include reference results, mostly obtained by simulation as well measurements from a live Spark cluster, that verify the tightness of our performance bounds.

The remainder of this paper is structured as follows. In Sec.~\ref{sec:forkjoin}, we formulate basic models of G$\mid$G$\mid$1 as well as GI$\mid$GI$\mid$1 fork-join systems in max-plus system theory. Multi-stage fork-join networks are considered in Sec.~\ref{sec:multistage}. We compare fork-join systems with classical multi-server systems with thinning and optional resequencing in Sec.~\ref{sec:thinning}. In Sec.~\ref{sec:nonidling}, we analyze non-idling single-queue implementations of multi-server and fork-join systems, respectively. Sec.~\ref{sec:conclusions} presents brief conclusions. Extensive proofs and a detailed description of the simulation and Spark experiments are in the appendix.
%
%
\section{Basic Fork-Join Systems}
\label{sec:forkjoin}
In this section, we derive a set of results for basic fork-join systems in max-plus system theory~\cite{baccelli:synchronizationlinearity, chang:performanceguarantees, jiang:maxplus, jiang:onecoin, luebben:availbw2}. Max-plus system theory is a branch of the deterministic~\cite{cruz:networkdelaycalculus, chang:performanceguarantees, leboudec:networkcalculus}, respectively, stochastic network calculus~\cite{chang:performanceguarantees, burchard:endtoendstatisticalcalculus, ciucu:networkservicecurvescaling2, fidler:momentcalculus, jiang:stochasticnetworkcalculus, fidler:netcalcsurvey, ciucu:goodvalue, fidler:netcalcguide}. In comparison to~\cite{rizk:forkjoin}, which is focused entirely on waiting and sojourn times of specific systems, the more general max-plus approach enables us to construct multi-stage fork-join networks as well as more advanced fork-join systems. Further, we generalize central results from~\cite{rizk:forkjoin} considering general arrival and service processes. Throughout this work, we consider only the case of homogeneous servers, i.e., all servers have identical service time distribution. Heterogeneous servers can be dealt with in the same way by a notational extension. We show results for heterogeneous servers and load balancing in~\cite{fidler:forkjoin}.
\subsection{Notation and Queueing Model}
We label jobs in the order of arrival by $n \ge 1$ and let $A(n)$ denote the time of arrival of job $n$. It follows for $n \ge m \ge 1$ that $A(n) \ge A(m) \ge 0$. For notational convenience, we define $A(0) = 0$. Further, we let $A(m,n) = A(n) - A(m)$ be the time between the arrival of job $m$ and job $n$ for $n \ge m \ge 1$. Hence, $A(n,n+1)$ is the inter-arrival time between job $n$ and job $n+1$ for $n \ge 1$. Similarly, $D(n)$ denotes departure times. To model systems, we adapt the definition of g-server from~\cite[Def. 6.3.1]{chang:performanceguarantees} using a notion of service process $S(m,n)$.
\begin{definition}[Max-plus server]
\label{def:maxplusserviceprocess}
A system with arrivals $A(n)$ and departures $D(n)$ is an $S(m,n)$ server under the max-plus algebra if it holds for all $n \ge 1$ that
\begin{equation*}
D(n) \le \max_{m \in [1,n]} \{ A(m) + S(m,n) \} .
\end{equation*}
It is an exact $S(m,n)$ server if it holds for all $n \ge 1$ that
\begin{equation*}
D(n) = \max_{m \in [1,n]} \{ A(m) + S(m,n) \} .
\end{equation*}
\end{definition}

The following Lem.~\ref{lem:exactmaxplusserviceprocess} shows that the general class of work-conserving systems satisfy the definition of exact server. We use $V(n)$ to denote the time at which job $n$ starts service.
\begin{lemma}[Work-conserving system]
\label{lem:exactmaxplusserviceprocess}
Consider a lossless, work-conserving, first-in first-out system and let $L(n)$ denote the service time of job $n$, where $n \ge 1$. Define for $n \ge m \ge 1$
\begin{equation*}
S(m,n) = \sum_{\nu=m}^n L(\nu) .
\end{equation*}
The system is an exact $S(m,n)$ server.
\end{lemma}
\begin{proof}
Since the system is lossless, work-conserving, and serves jobs in first-in first-out order, job $n \ge 2$ starts service at
\begin{equation}
V(n) = \max \{ A(n), V(n-1) + L(n-1) \} ,
\label{eq:starttime}
\end{equation}
and job 1 at $V(1) = A(1)$. By recursive insertion of~\eqref{eq:starttime} we have
\begin{equation}
V(n) = \max_{m \in [1,n]} \left\{ A(m) + \sum_{\nu=m}^{n-1} L(\nu) \right\} ,
\label{eq:starttimesolved}
\end{equation}
for $n \ge 1$. Since $D(n) = V(n) + L(n)$, it follows with~\eqref{eq:starttimesolved} that $D(n) = \max_{m \in [1,n]} \{A(m) + \sum_{\nu=m}^{n} L(\nu)\}$, which proves that the work-conserving system is an exact max-plus server.
\end{proof}

For the sojourn time of job $n \ge 1$, defined as $T(n) = D(n) - A(n)$, it follows by insertion of Def.~\ref{def:maxplusserviceprocess} that
\begin{equation}
T(n) = \max_{m \in [1,n]} \{ S(m,n) - A(m,n) \} .
\label{eq:sojourntime}
\end{equation}
The waiting time of job $n \ge 1$ is $W(n) = V(n) - A(n)$. As in the case of work-conserving systems in Lem.~\ref{lem:exactmaxplusserviceprocess}, $V(n) = \max \{A(n), D(n-1)\}$, so we have $W(n) = [D(n-1) - A(n)]^+$, where $[X]^+ = \max\{X,0\}$ is the non-negative part and $D(0)=0$ by definition. With Def.~\ref{def:maxplusserviceprocess}, it holds that
\begin{equation}
W(n) = \biggl[ \sup_{m \in [1,n-1]} \{ S(m,n-1) - A(m,n) \} \biggr]^+ .
\label{eq:waitingtime}
\end{equation}
Here, we use the supremum since for $n=1$ \eqref{eq:waitingtime} evaluates to an empty set. For non-negative real numbers the $\sup$ of an empty set is zero. While we used the definition of an exact server to derive an expression for the sojourn and waiting times, we note that the upper bound specified by the definition of a server is usually sufficient, as it provides upper bounds of sojourn and waiting times.
%
%
\subsection{Statistical Performance Bounds}
Next, we derive statistical performance bounds for servers as defined above. Throughout the paper, we generally assume that the arrival and service processes are independent of each other. Considering general arrival and service processes, the server is a G$\mid$G$\mid$1 queue. The results enable us to generalize recent findings obtained for iid service times, i.e., for a GI service model, in~\cite{rizk:forkjoin}.

We consider arrival and service processes that belong to the broad class of $(\sigma,\rho)$-constrained processes~\cite{chang:performanceguarantees}, that are characterized by affine bounding functions of the moment generating function (MGF). The MGF of a random variable $X$ is defined as $\mathsf{M}_X(\theta) =\mathsf{E}\bigl[e^{\theta X}\bigr]$ where $\theta$ is a free parameter. The following definition adapts~\cite{chang:performanceguarantees} to max-plus systems.
\begin{definition}
\label{def:sigmarho}
An arrival process is $(\sigma_A,\rho_A)$-lower constrained if for all $n \ge m \ge 1$ and $\theta > 0$ it holds that
\begin{equation*}
\mathsf{E}\Bigl[e^{-\theta A(m,n)}\Bigr] \le e^{-\theta (\rho_A(-\theta) (n-m) - \sigma_A(-\theta))} .
\end{equation*}
Similarly, a service process is $(\sigma_S,\rho_S)$-upper constrained if for all $n \ge m \ge 1$ and $\theta > 0$ it holds that
\begin{equation*}
\mathsf{E}\Bigl[e^{\theta S(m,n)}\Bigr] \le e^{\theta(\rho_S(\theta) (n-m+1) + \sigma_S(\theta))} .
\end{equation*}
\end{definition}
Considering the service times of jobs as in Lem.~\ref{lem:exactmaxplusserviceprocess}, we also apply Def.~\ref{def:sigmarho} to characterize the cumulative service process $L(m,n) = \sum_{\nu=m}^n L(\nu)$ by $(\sigma_L,\rho_L)$.

In the special case of GI arrival processes, $A(m,n) = \sum_{\nu=m}^{n-1} A(\nu,\nu+1)$ has iid inter-arrival times $A(\nu,\nu+1)$. It follows that $\mathsf{E}\bigl[e^{-\theta A(\nu,\nu+1)}\bigr] = \mathsf{E}\bigl[e^{-\theta A(1,2)}\bigr]$ for $\nu \ge 1$. Next, we use that the MGF of a sum of independent random variables is the product of their individual MGFs, i.e., $\mathsf{M}_{X+Y}(\theta) = \mathsf{M}_{X}(\theta)\mathsf{M}_{Y}(\theta)$ to derive minimal traffic parameters from Def.~\ref{def:sigmarho} as $\sigma_A(-\theta) = 0$ and
\begin{equation}
\rho_A(-\theta) = -\frac{1}{\theta} \ln \mathsf{E}\Bigl[e^{-\theta A(1,2)}\Bigr] .
\label{eq:arrivalparameter}
\end{equation}
Similarly for GI service processes, $S(m,n) = \sum_{\nu=m}^n S(\nu)$ is composed of iid service increments $S(\nu)$ that have minimal parameters $\sigma_S(\theta) = 0$ and
\begin{equation}
\rho_S(\theta) = \frac{1}{\theta} \ln \mathsf{E}\Bigl[e^{\theta S(1)}\Bigr].
\label{eq:serviceparameter}
\end{equation}
Parameter $\rho_A(-\theta)$ decreases with $\theta > 0$ from the mean to the minimum inter-arrival time and $\rho_S(\theta)$ increases with $\theta > 0$ from the mean to the maximum service time.
\begin{theorem}[Statistical performance bounds]
\label{th:gg1}
Consider a server as in Def.~\ref{def:maxplusserviceprocess}, with arrival and service parameters $(\sigma_A(-\theta),\rho_A(-\theta))$ and $(\sigma_S(\theta),\rho_S(\theta))$ as specified by Def.~\ref{def:sigmarho}. For $n \ge 1$, the sojourn time $T(n) = D(n) - A(n)$ satisfies
\begin{equation*}
\mathsf{P} [ T(n) > \tau ] \le \alpha e^{\theta \rho_{S}(\theta)} e^{-\theta \tau} ,
\end{equation*}
and the waiting time $W(n) = [D(n-1) - A(n)]^+$ satisfies
\begin{equation*}
\mathsf{P} \left[ W(n) > \tau \right] \le \alpha e^{-\theta \tau} .
\end{equation*}
In the case of G$\mid$G arrival and service processes, the free parameter $\theta > 0$ has to satisfy $\rho_{S}(\theta) < \rho_A(-\theta)$ and
\begin{equation*}
\alpha = \frac{e^{\theta(\sigma_A(-\theta) + \sigma_{S}(\theta))}}{1-e^{-\theta (\rho_A(-\theta)-\rho_{S}(\theta))}} .
\end{equation*}
In the special case of GI$\mid$GI arrival and service processes, $\theta > 0$ has to satisfy $\rho_{S}(\theta) \le \rho_A(-\theta)$ and $\alpha=1$.
\end{theorem}
\begin{figure}
  \centering
  \includegraphics[width=0.75\columnwidth]{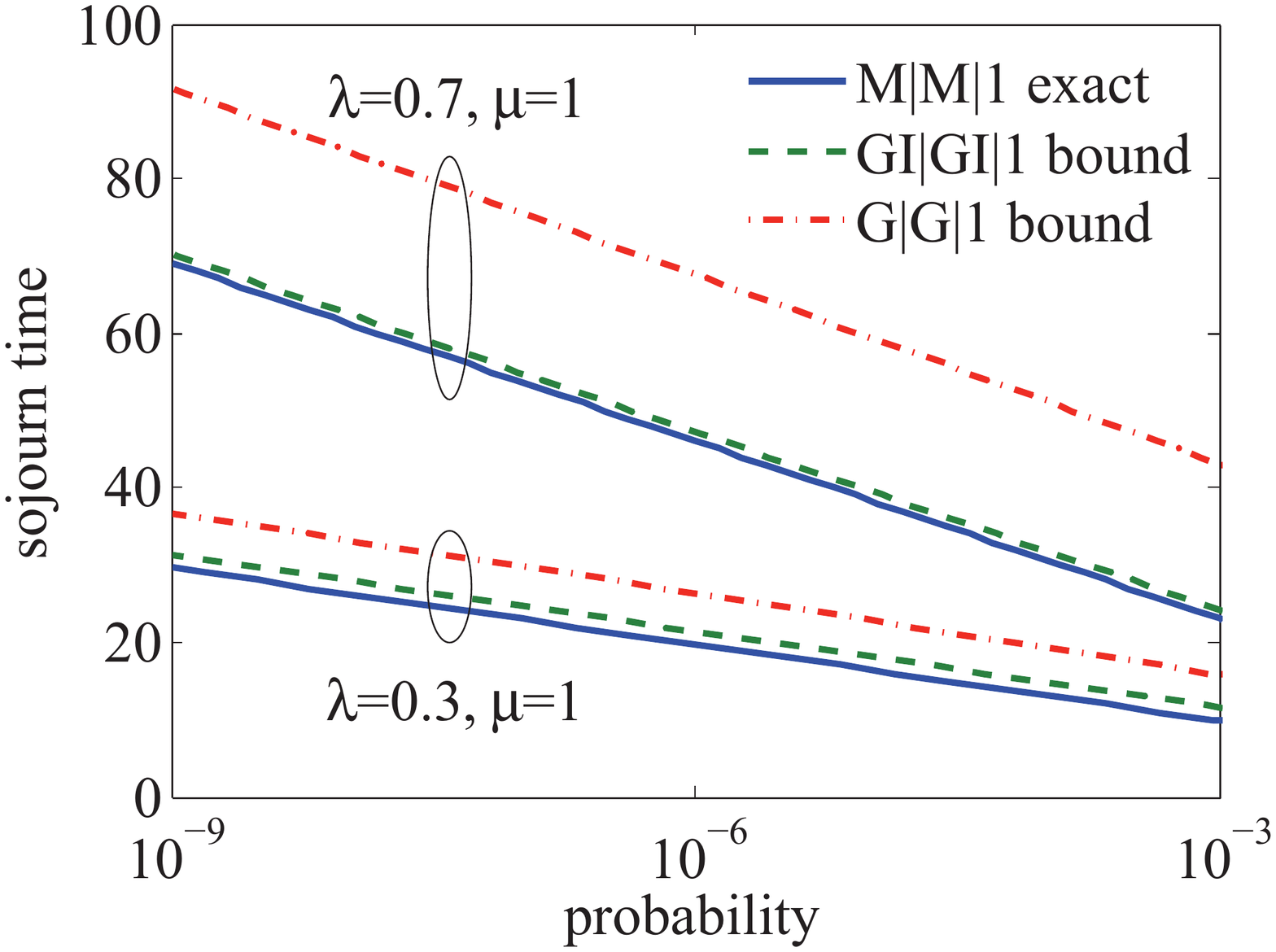}
  \caption{M$\mid$M$\mid$1 queue. The bounds show the correct exponential tail decay.}
  \label{fig:backlog_qt}
\end{figure}
The proof is provided in the appendix. For the special case of GI$\mid$GI arrival and service processes, Th.~\ref{th:gg1} recovers the classical bound for the waiting time of GI$\mid$GI$\mid$1 queues~\cite{kingman:gg1} in the max-plus system theory. Like~\cite{kingman:gg1}, the proof uses Doob's martingale inequality~\cite{doob:stochasticprocesses}. The proof for the G$\mid$G arrival and service processes adapts the approach from~\cite{chang:performanceguarantees, fidler:momentcalculus} to max-plus systems. The important property of the G$\mid$G result is that it differs only by a constant factor $\alpha$ from the GI$\mid$GI result and otherwise recovers the characteristic exponential tail decay $e^{-\theta \tau}$ with the same maximal decay rate $\theta$.
\paragraph*{M$\mid$M$\mid$1 Queue}
For evaluation of the accuracy of the bounds in Th.~\ref{th:gg1}, we consider the basic case of an M$\mid$M$\mid$1 queue, where exact results are available for comparison. Given iid exponential inter-arrival and service times with parameters $\lambda$ and $\mu$, respectively, \eqref{eq:arrivalparameter} and \eqref{eq:serviceparameter} evaluate to
\begin{equation}
\rho_A(-\theta) = -\frac{1}{\theta} \ln \left(\frac{\lambda}{\lambda+\theta}\right) ,
\label{eq:expoarrivalparameter}
\end{equation}
for $\theta > 0$, and
\begin{equation}
\rho_{S}(\theta) = \frac{1}{\theta} \ln \left(\frac{\mu}{\mu-\theta}\right) ,
\label{eq:exposerviceparameter}
\end{equation}
for $\theta \in (0,\mu)$. From the condition $\rho_{S}(\theta) \le \rho_A(-\theta)$ it follows that $\theta \le \mu-\lambda$ under the stability condition $\lambda < \mu$. By the choice of the maximal $\theta = \mu - \lambda$ we have from Th.~\eqref{th:gg1} that
\begin{equation}
\mathsf{P} \left[ T(n) > \tau \right] \le \frac{\mu}{\lambda} e^{-(\mu-\lambda) \tau} .
\label{eq:mm1bound}
\end{equation}

Compared to the exact distribution of the sojourn time of the M$\mid$M$\mid$1 queue, that is $\mathsf{P} \left[ T(n) > \tau \right] = e^{-(\mu-\lambda) \tau}$ see, e.g.,~\cite{adan:queueingsystems}, \eqref{eq:mm1bound} has the same tail decay and differs only by the pre-factor $\mu/\lambda$. Obviously, the bound becomes tighter if the utilization is high, in which case $\mu/\lambda$ approaches one.

In Fig.~\ref{fig:backlog_qt}, we illustrate the bounds from Th.~\ref{th:gg1} compared to the exact M$\mid$M$\mid$1 result. Clearly, the curves show the same tail decay, where the GI$\mid$GI$\mid$1 bound provides better numerical accuracy compared to the G$\mid$G$\mid$1 bound that does not use independence of the increment processes and hence has parameter $\alpha > 1$. In the case of the G$\mid$G$\mid$1 bound, the parameter $\theta$ is optimized numerically to obtain the smallest delay bound.

For numerical evaluation, we will frequently use M$\mid$M arrival and service processes as specified by \eqref{eq:expoarrivalparameter} and \eqref{eq:exposerviceparameter}. We note that Th.~\ref{th:gg1} provides results for G$\mid$G arrival and service processes by substitution of the MGFs of the respective processes into Def.~\ref{def:sigmarho}.
\subsection{Fork-Join Systems}
\label{sec:forkjoinidling}
In a fork-join system, each job $n \ge 1$ is composed of $k$ tasks with service times $Q_i(n)$ for $i \in [1,k]$, i.e., the service requirements of the tasks may differ from each other and may or may not be independent. The tasks are distributed to $k$ parallel servers (fork) and once all tasks of a job are served, the job leaves the system (join), see Fig.~\ref{fig:fjqueue}. The parallel servers are not synchronized; i.e., server $i$ starts serving task $i$ of job $n+1$ (assuming it is already in the system), once it finishes serving task $i$ of job $n$, which departs from server $i$ at $D_i(n)$. Job $n$ has finished service once all of its tasks $i \in [1,k]$ have finished service. The following lemma shows that fork-join systems are servers under the max-plus algebra. After estimating the MGF of the respective service process, performance bounds are obtained.
\begin{lemma}[Fork-join system]
\label{lem:forkjoin}
Consider a fork-join system with $k$ parallel servers as in Lem.~\ref{lem:exactmaxplusserviceprocess}. Let $Q_i(n)$ denote the service time of task $i$ of job $n$ where $i \in [1,k]$ and $n \ge 1$. Define for $n \ge m \ge 1$
\begin{equation*}
S(m,n) = \max_{i \in [1,k]} \Biggl\{ \sum_{\nu=m}^n Q_i(\nu) \Biggr\} .
\end{equation*}
The system is an exact $S(m,n)$ server.
\end{lemma}
\begin{proof}
Since a job departs from the system once all of its tasks $i \in [1,k]$ are completed, we have for $n \ge 1$ that
\begin{equation}
D(n) = \max_{i \in [1,k]} \{ D_i(n) \} .
\label{eq:departuresforkjoin}
\end{equation}
By insertion of Def.~\ref{def:maxplusserviceprocess} for each of the servers $i \in [1,k]$, it follows that
\begin{equation*}
D(n) = \max_{i \in [1,k]} \left\{ \max_{m \in [1,n]} \{A(m) + S_i(m,n) \} \right\}  .
\end{equation*}
After reordering the maxima
\begin{equation*}
D(n) = \max_{m \in [1,n]} \left\{A(m) + \max_{i \in [1,k]} \{ S_i(m,n) \} \right\} ,
\end{equation*}
we conclude that the fork-join system is an exact $S(m,n) = \max_{i \in [1,k]} \{ S_i(m,n) \}$ server. In the last step, we invoke Lem.~\ref{lem:exactmaxplusserviceprocess} with $Q_i(n)$ for each of the servers $i \in [1,k]$.
\end{proof}
Next, we estimate the MGF of the service process $S(m,n)$ from Lem~\ref{lem:forkjoin} for $n \ge m \ge 1$ by
\begin{equation*}
\mathsf{E}\Bigl[e^{\theta S(m,n)}\Bigr] \le \sum_{i=1}^k \mathsf{E}\Bigl[e^{\theta \sum_{\nu=m}^n Q_i(\nu)}\Bigr] .
\end{equation*}
Assuming homogeneous tasks with parameters $(\sigma_{Q}(\theta),\rho_{Q}(\theta))$ for $i \in [1,k]$, it follows by insertion of Def.~\ref{def:sigmarho} that
\begin{align*}
\mathsf{E}\Bigr[e^{\theta S(m,n)}\Bigl] \le&  k e^{\theta (\sigma_{Q}(\theta) + \rho_{Q}(\theta) (n-m+1))} .
\end{align*}
This shows that the service process of the fork-join system has parameters
\begin{equation}
\sigma_S(\theta) = \sigma_{Q}(\theta) + \frac{\ln k}{\theta} ,
\label{eq:forkjoinsigma}
\end{equation}
and
\begin{equation}
\rho_S(\theta) = \rho_{Q}(\theta) .
\label{eq:forkjoinrho}
\end{equation}
Performance bounds follow as a corollary of Th.~\ref{th:gg1}.
\begin{corollary}[Fork-join system]
\label{cor:forkjoin}
Consider a fork-join system as in Lem.~\ref{lem:forkjoin}, with arrival and service parameters $(\sigma_A(-\theta),\rho_A(-\theta))$ and $(\sigma_{Q}(\theta),\rho_{Q}(\theta))$ as specified by Def.~\ref{def:sigmarho}. For $n \ge 1$, the sojourn time satisfies
\begin{equation*}
\mathsf{P}[ T(n) > \tau ] \le k \alpha e^{\theta \rho_{Q}(\theta)} e^{-\theta\tau} ,
\end{equation*}
and the waiting time of the task that starts service last
\begin{equation*}
\mathsf{P} \left[ W(n) > \tau \right] \le k \alpha e^{-\theta \tau} .
\end{equation*}
In the case of G$\mid$G arrival and service processes, the free parameter $\theta > 0$ has to satisfy $\rho_{Q}(\theta) < \rho_A(-\theta)$ and
\begin{equation*}
\alpha = \frac{e^{\theta(\sigma_A(-\theta) + \sigma_{Q}(\theta))}}{1-e^{-\theta (\rho_A(-\theta)-\rho_{Q}(\theta))}} .
\end{equation*}
In the special case of GI$\mid$GI arrival and service processes, $\theta > 0$ has to satisfy $\rho_{Q}(\theta) \le \rho_A(-\theta)$ and $\alpha=1$.
\end{corollary}
\begin{proof}
For G$\mid$G arrival and service processes, Cor.~\ref{cor:forkjoin} is obtained directly by insertion of \eqref{eq:forkjoinsigma} and \eqref{eq:forkjoinrho} into Th.~\ref{th:gg1}. We note that the waiting time of a job is defined to be that of its task that starts service last. This follows by insertion of \eqref{eq:departuresforkjoin} into the definition of waiting time $W(n) = [D(n-1) - A(n)]^+$.

As the increment process of $S(m,n)$ in Lem~\ref{lem:forkjoin} is non-trivial, we pursue a different approach\footnote{The approach applies also in the case of G$\mid$G arrival and service processes. We showed the alternative approach via \eqref{eq:forkjoinsigma} and \eqref{eq:forkjoinrho} nevertheless, as it extends to multi-stage fork-join networks, see Sec.~\ref{sec:multistage}.} to show the result for the special case of GI$\mid$GI arrival and service processes. With \eqref{eq:departuresforkjoin}, we derive the sojourn time $T(n) = D(n) - A(n)$ of job $n$ as $T(n) = \max_{i \in [1,k]} \{ D_i(n) - A(n) \}$ for $n \ge 1$. Hence, the sojourn time of job $n$ is expressed as a maximum $T(n) = \max_{i \in [1,k]} \{T_i(n) \}$ of the sojourn times $T_i(n) = D_i(n) - A(n)$ of the individual tasks $i \in [1,k]$ of job $n$. Since the individual servers satisfy Lem.~\ref{lem:exactmaxplusserviceprocess}, we can invoke Th.~\ref{th:gg1} for each of the servers $i \in [1,k]$ and use the union bound to obtain the result of Cor.~\ref{cor:forkjoin}. The waiting time $\max_{i \in [1,k]} \{ W_i(n) \}$ where $W_i(n)$ as in \eqref{eq:waitingtime} can be derived in the same way.
\end{proof}
We note that Cor.~\ref{cor:forkjoin} does not make an assumption of independence regarding the parallel servers. Indeed, independence cannot be assumed as the waiting and sojourn times of the individual servers depend on the same arrival process~\cite{baccelli:forkjoin, kemper:forkjoin}.

To investigate the scaling of fork-join systems with $k$ parallel servers, we first note that the stability condition $\rho_{Q}(\theta) < \rho_A(-\theta)$ does not depend on $k$. Hence, the maximal speed of the tail decay of the performance bounds $\theta$ is independent of $k$. Next, we equate the sojourn time bound from Cor.~\ref{cor:forkjoin} with $\varepsilon$ and solve for
\begin{equation}
\tau \le \rho_{Q}(\theta) + \frac{\ln k + \ln\alpha - \ln\varepsilon}{\theta} ,
\label{eq:forkjoinlogkgrowth}
\end{equation}
for $\theta > 0$ subject to the stability condition $\rho_{Q}(\theta) < \rho_A(-\theta)$. Eq.~\eqref{eq:forkjoinlogkgrowth} expresses a sojourn time bound that is exceeded at most with probability $\varepsilon$. It exhibits a growth with $\ln k$. The growth is larger for smaller $\theta$ corresponding to a higher utilization.

We include an example that demonstrates the quick estimation of the expected value from the sojourn time bound. By integration of the tail of the sojourn time from Cor.~\ref{cor:forkjoin} we have
\begin{align*}
\mathsf{E}[T(n)] = & \int_{0}^{\infty} \mathsf{P} [ T(n) > \tau ] \mathrm{d}\tau \\
\le & \int_{0}^{\tau^*} \mathrm{d}\tau + k \alpha e^{\theta \rho_{Q}(\theta)} \int_{\tau^*}^{\infty} e^{-\theta \tau} \mathrm{d}\tau,
\end{align*}
where we used that $\mathsf{P} [ T(n) > \tau ] \le 1$ to determine $\tau^* = \ln \bigl(k \alpha e^{\theta \rho_{Q}(\theta)}\bigr)/\theta$. It follows that the expected sojourn time
\begin{equation}
\mathsf{E} [ T(n) ] \le \rho_{Q}(\theta) +  \frac{\ln k + \ln \alpha + 1}{\theta}
\label{eq:homogeneousparallelexpectedsojourntime}
\end{equation}
is also limited by $\ln k$. The result applies for general arrival and service processes and generalizes the finding of $\ln k$ that is obtained in~\cite{baccelli:forkjoin, rizk:forkjoin} for iid service times.
\paragraph*{M$\mid$M tasks}
\begin{figure}
\centering
\subfigure[$(1-\varepsilon)$-quantile] {
\includegraphics[width=0.465\linewidth]{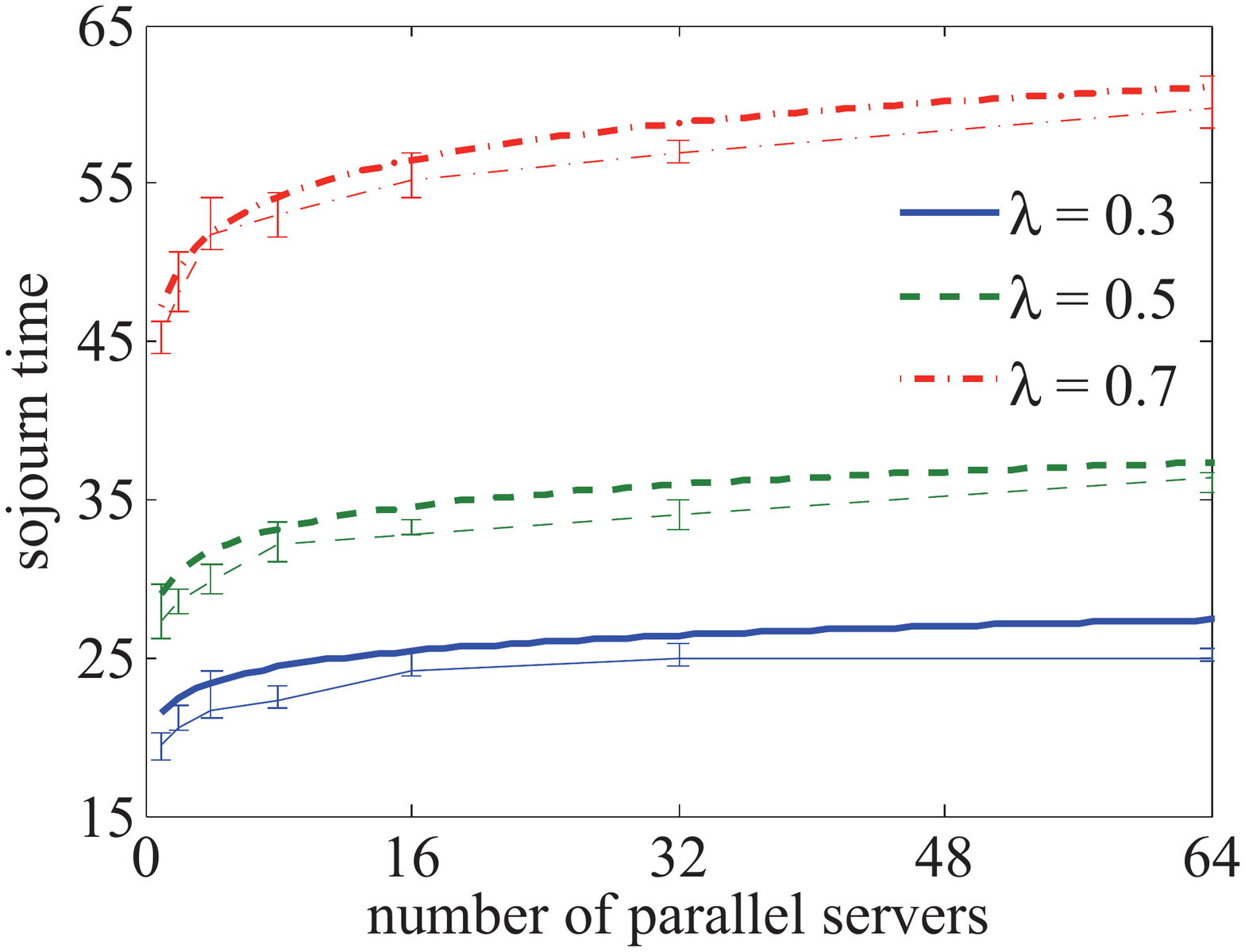}
}
\hfill
\subfigure[expected value] {
\includegraphics[width=0.465\linewidth]{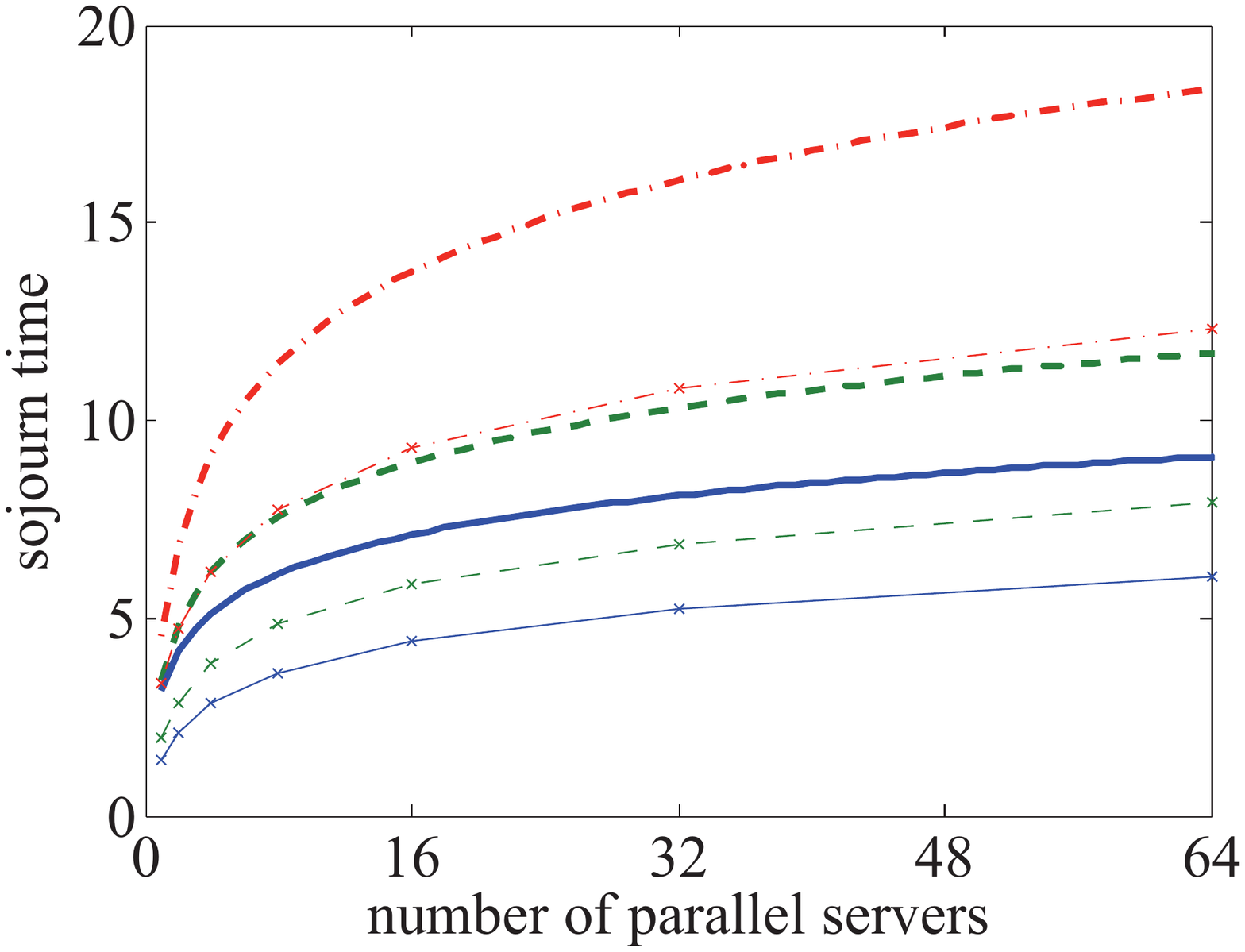}
}
\caption{Fork-join system. Analytical bounds (thick lines) and simulation results (thin lines). Sojourn times grow with $\ln k$ for $k$ servers.}
\label{fig:mm1forkjoin}
\end{figure}
In Fig.~\ref{fig:mm1forkjoin}, we consider a fork-join system with $k \ge 1$ parallel servers. Jobs have iid exponential inter-arrival times and are composed of $k$ tasks with iid exponential service times each. The parameters $\rho_A(-\theta)$ and $\rho_{Q}(\theta)$ for $i \in [1,k]$ are as specified by \eqref{eq:expoarrivalparameter} and \eqref{eq:exposerviceparameter} where we let $\mu=1$. We show bounds of the expected sojourn time and sojourn time quantiles $\tau$, where $\mathsf{P}[T(n) > \tau] \le \varepsilon$ and $\varepsilon = 10^{-6}$. The curves show the characteristic logarithmic growth with $k$. This is also confirmed in simulation results that agree well with the sojourn time bounds. The expectation of the sojourn time~\eqref{eq:homogeneousparallelexpectedsojourntime} is only a rough estimate, as anticipated.
\subsection{Split-Merge Systems}
Split-merge systems, see Fig.~\ref{fig:smqueue}, are a variant of fork-join systems where all tasks of a job have to start execution simultaneously. If a server $i$ finishes task $i$ of job $n$, it idles until all tasks $j \in [1,k]$ of that job are finished before any of the tasks of job $n+1$ if any starts.
\begin{lemma}[Split-merge system]
\label{lem:splitmerge}
Consider a split-merge system with $k$ parallel servers as in Lem.~\ref{lem:exactmaxplusserviceprocess}. Let $Q_i(n)$ denote the service time of task $i$ of job $n$ where $i \in [1,k]$ and $n \ge 1$. Define for $n \ge m \ge 1$
\begin{equation*}
S(m,n) = \sum_{\nu=m}^n \max_{i \in [1,k]} \{ Q_i(\nu) \} .
\end{equation*}
The system is an exact $S(m,n)$ server.
\end{lemma}
\begin{figure}
  \centering
  \includegraphics[width=0.77\columnwidth]{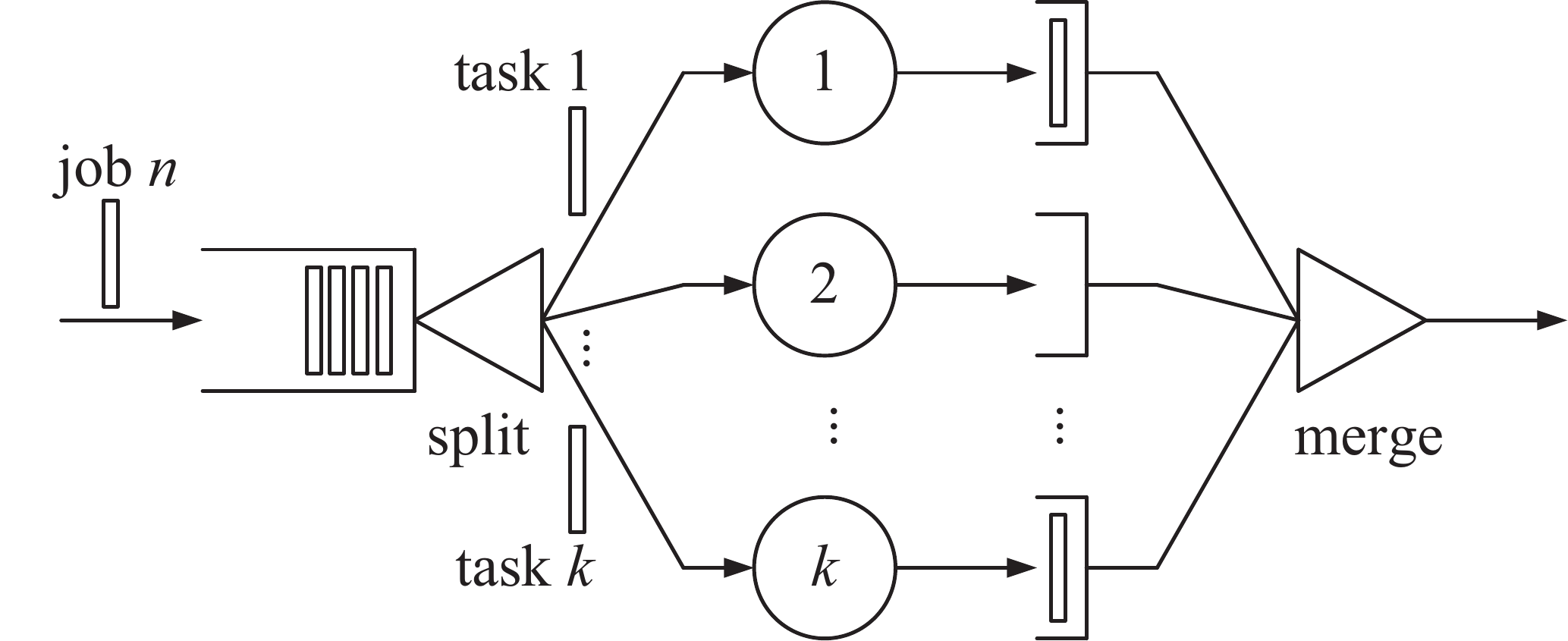}
  \caption{Split-merge system. Compared to the fork-join system, tasks have an additional synchronization constraint, i.e., the execution of the tasks of a job has to start at the same time.}
  \label{fig:smqueue}
\end{figure}
\begin{proof}
Since all tasks $i \in [1,k]$ of job $n$ start service simultaneously at $V(n)$ and only after all tasks of job $n-1$ have finished service, we have for $n \ge 2$ that
\begin{equation}
V(n) = \max \biggl\{A(n), V(n-1) + \max_{i \in [1,k]} \{Q_i(n-1)\}\biggr\},
\label{eq:starttimesplitmerge}
\end{equation}
Recursive insertion of \eqref{eq:starttimesplitmerge} with $V(1) = A(1)$ as in Lem.~\ref{lem:exactmaxplusserviceprocess}, and letting $D(n) = V(n) + \max_{i \in [1,k]} \{Q_i(n)\}$ completes the proof.
\end{proof}
Since Lem.~\ref{lem:splitmerge} proves that the split-merge system satisfies the definition of server Def.~\ref{def:maxplusserviceprocess}, the performance bounds of Th.~\ref{th:gg1} apply, using the parameters of the service process $S(m,n)$ in Lem.~\ref{lem:splitmerge}. In the case of iid service times, the service parameters of Lem.~\ref{lem:splitmerge} are derived from \eqref{eq:serviceparameter} as
\begin{equation}
\rho_S(\theta) = \frac{1}{\theta} \ln \mathsf{E}\Bigl[e^{\theta \max_{i \in [1,k]} \{Q_i(1)\}}\Bigr] .
\label{eq:rhosplitmergeiid}
\end{equation}

The general problem of split-merge systems is, however, that $\max_{i \in [1,k]} \{ Q_i(\nu) \}$ is stochastically increasing with $k$, with few exceptions such as in the case of identical task service times. The increase implies longer idle times that result in a reduced stability region. As a quick estimate of \eqref{eq:rhosplitmergeiid}
\begin{equation*}
\rho_S(\theta) \le \frac{1}{\theta} \ln \Bigl(k \mathsf{E}\Bigl[e^{\theta Q_1(1)}\Bigr] \Bigr)
\end{equation*}
shows that $\rho_S$ has at most a logarithmic growth with the number of parallel servers, resulting in a corresponding reduction of the stability region. A decrease of the stability region with $\ln k$ is also shown in~\cite{rizk:forkjoin}, where the authors advise against split-merge implementations based on an in-depth comparison with fork-join systems.
\subsection{Replication Systems}
Given $k$ parallel servers, one option to deal with stragglers is the redundant execution of $k$ replicated jobs. In this case, the service time of a job is determined as the minimum of the service times of all its replicas. Once a replica has finished service, all other replicas of the job may or may not be purged. We consider non-purging $(k,l)$ fork-join systems in~\cite{fidler:forkjoin}, which includes pure replication systems as a special case if $l=1$. Systems with replication and purging are elaborated on in~\cite{poloczek:parallelsystems}. The following lemma says that replication systems with purging satisfy Def.~\ref{def:maxplusserviceprocess}. This basic property implies that the performance bounds of Th.~\ref{th:gg1} hold.
\begin{lemma}[Replication system]
\label{lem:replication}
Consider a purging replication system with $k$ parallel servers as in Lem.~\ref{lem:exactmaxplusserviceprocess}. Let $L_i(n)$ denote the service time of replica $i$ of job $n$ where $i \in [1,k]$ and $n \ge 1$. For $n \ge m \ge 1$ define
\begin{equation*}
S(m,n) = \sum_{\nu=m}^n \min_{i \in [1,k]} \{ L_i(\nu) \} .
\end{equation*}
The system is an exact $S(m,n)$ server.
\end{lemma}
\begin{proof}
First, we note that all replicas of job $n \ge 1$ start service at the same time $V(n)$. This is an immediate consequence of purging all replicas of a job once one replica finishes service. It follows for $n \ge 2$ that
\begin{equation}
V(n) = \max \biggl\{A(n), V(n-1) + \min_{i \in [1,k]} \{L_i(n-1)\}\biggr\},
\label{eq:starttimereplication}
\end{equation}
Recursive insertion of \eqref{eq:starttimereplication} with $V(1) = A(1)$ as in Lem.~\ref{lem:exactmaxplusserviceprocess}, and letting $D(n) = V(n) + \min_{i \in [1,k]} \{L_i(n)\}$ completes the proof.
\end{proof}
Performance bounds follow by insertion of the service parameters of $S(m,n)$ as defined by Lem.~\ref{lem:replication} into Th.~\ref{th:gg1}. As a detailed evaluation of replication systems with purging and correlated replicas is provided by~\cite{poloczek:parallelsystems}, we only evaluate the simple case of iid replicas with exponential service times $L_i$ with parameter $\mu$. It follows that $\min_{i \in [1,k]} \{L_i(\nu)\}$ is exponential with parameter $k\mu$ so that $\rho_S(\theta)$ follows by substitution of $k\mu$ for $\mu$ in~\eqref{eq:exposerviceparameter}.
%
%
\section{Multi-Stage Fork-Join Networks}
\label{sec:multistage}
We contribute a new bound on the growth of end-to-end sojourn times for multi-stage fork-join networks, where we consider $h$ fork-join stages\footnote{While we focus on multi-stage fork-join networks only, we note that the same analysis applies to networks of split-merge or replication systems.} in tandem, each with $k$ parallel servers. We use subscript $i \in [1,k]$ to distinguish the servers of a stage and superscript $j \in [1,h]$ to denote the stages. Note that jobs depart from each fork-join stage in the order of their arrival. The following lemma reproduces a fundamental result of the network calculus~\cite{chang:performanceguarantees}.
\begin{lemma}[Multi-stage fork-join network]
\label{lem:multistage}
Consider a multi-stage network of $h$ fork-join systems as in Lem.~\ref{lem:forkjoin} in tandem. Define for $n \ge m \ge 1$
\begin{multline*}
S^{\text{net}}(m,n) = \max_{\nu^j: m \le \nu^1\le\nu^2\le\dots\nu^{h-1}\le n}  \\ \{ S^1(m,\nu^1) + S^2(\nu^1,\nu^2) + \dots + S^h(\nu^{h-1},n) \} .
\end{multline*}
The fork-join network is an exact $S^{\text{net}}(m,n)$ server.
\end{lemma}
\begin{proof}
In a tandem of fork-join systems, the departures of stage $j$ are the arrivals of stage $j+1$, i.e., $A^{j+1}(n) = D^{j}(n)$ for $j \in [1,h-1]$. Further, since jobs depart from a fork-join system in the order of their arrival, we have for all $j \in [1,h]$ that $A^{j}(n) \ge A^{j}(m) \ge 0$ for $n \ge m \ge 1$. Next, we use that each fork-join stage is an exact server as in Def.~\ref{def:maxplusserviceprocess}. We start with $D^h(n) = \max_{\nu^{h-1} \in [1,n]} \{D^{h-1}(\nu^{h-1}) + S^{h}(\nu^{h-1},n) \}$ and recursively insert Def.~\ref{def:maxplusserviceprocess} for $j \in [1,h-1]$, to obtain
\begin{multline*}
D^h(n) = \max_{m\in[1,n]} \Bigl\{ \max_{\nu^j : m \le \nu^{1} \le \dots \le \nu^{h-1} \le n} \\ \bigl\{A^{1}(m) + S^{1} (m,\nu^{1}) + S^{2}(\nu^{1},\nu^{2}) + \dots + S^{h}(\nu^{h-1},n) \bigr\} \Bigr\},
\end{multline*}
for $n \ge 1$. This proves that $S^{\text{net}}(m,n)$ is an exact server.
\end{proof}
\begin{theorem}[Multi-stage fork-join network]
\label{th:multistage}
Consider a multi-stage fork-join network as in Lem.~\ref{lem:multistage} with arrival and service parameters $(\sigma_A(-\theta),\rho_A(-\theta))$ and $(\sigma_{Q}(\theta),\rho_{Q}(\theta))$ as specified by Def.~\ref{def:sigmarho}. Let the service times at each of the stages be independent. For $n \ge 1$, the end-to-end sojourn time satisfies
\begin{equation*}
\mathsf{P}[ T(n) > \tau ] \le k^h \alpha e^{\theta h \rho_{Q}(\theta)} e^{-\theta\tau} ,
\end{equation*}
where $\theta > 0$ has to satisfy $\rho_{Q}(\theta) < \rho_A(-\theta)$ and
\begin{equation*}
\alpha = \frac{e^{\theta(\sigma_A(-\theta) + h \sigma_{Q}(\theta))}}{\bigl(1-e^{-\theta (\rho_A(-\theta)-\rho_{Q}(\theta))}\bigr)^h} .
\end{equation*}
\end{theorem}
The proof can be found in the appendix.  It uses an established method of the stochastic network calculus~\cite{fidler:momentcalculus}; see also the tutorial~\cite{fidler:netcalcguide} or the textbook~\cite{jiang:basicstochasticcalculus}.

To see the growth of $\tau$ with $h$ and $k$, we equate the sojourn time bound in Th.~\ref{th:multistage} with $\varepsilon$ and solve for
\begin{multline*}
\tau = \sigma_A(-\theta) + h (\sigma_{Q}(\theta) + \rho_{Q}(\theta)) \\ + \frac{1}{\theta} \left(h \ln k - h \ln \left(1-e^{-\theta (\rho_A(-\theta) - \rho_{Q}(\theta))}\right) - \ln \varepsilon \right)
\end{multline*}
that grows in $\mathcal{O}(h \ln k)$. The result compares to a growth in $\mathcal{O}(h \ln (hk))$ obtained previously in~\cite{fidler:forkjoin}. The improvement is achieved by taking advantage of the statistical independence of the stages, whereas~\cite{fidler:forkjoin} does not make this assumption.
%
%
\paragraph*{M$\mid$M tasks}
Fig.~\ref{fig:multistage} shows sojourn time bounds for a multi-stage fork-join network with up to $h=64$ stages each with $k \in \{1,2,4,8,16\}$ parallel servers. The inter-arrival and service times are exponential with parameters $\lambda=0.5$ and $\mu=1$, respectively, and $\varepsilon = 10^{-6}$ as in Fig.~\ref{fig:mm1forkjoin}. The analytical results in Fig.~\ref{fig:multistageanalysis} exhibit the same characteristic trends as the simulation results in Fig.~\ref{fig:multistagesimulation}, albeit with less precision due to the inequalities that are involved for each stage. The end-to-end sojourn times show a linear growth with $h$ and a logarithmic growth with $k$, observable by the equidistantly spaced lines for $k \in \{1,2,4,8,16\}$.
\begin{figure}
\centering
\subfigure[analysis] {
\includegraphics[width=0.465\linewidth]{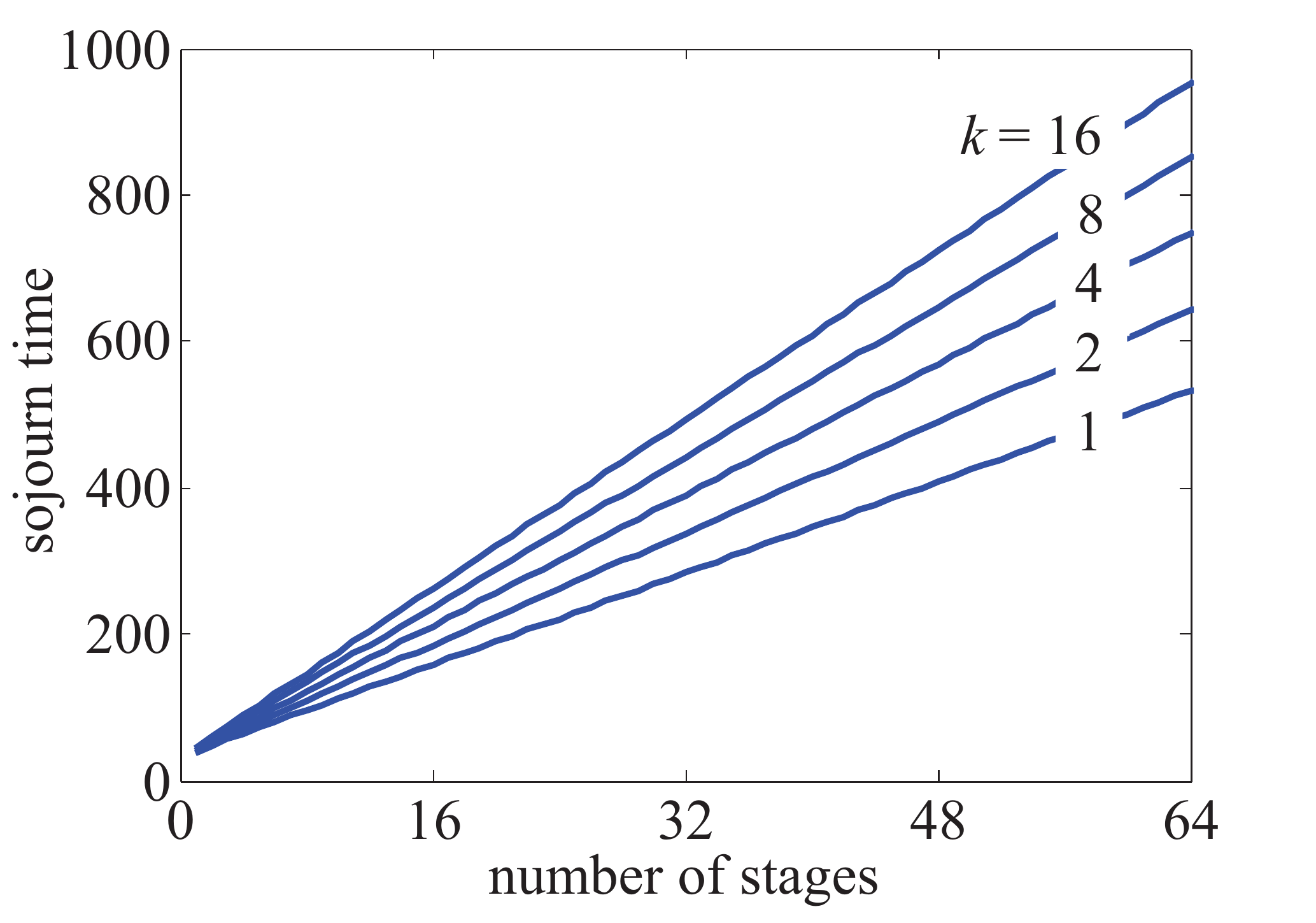}
\label{fig:multistageanalysis}
}
\hfill
\subfigure[simulation] {
\includegraphics[width=0.465\linewidth]{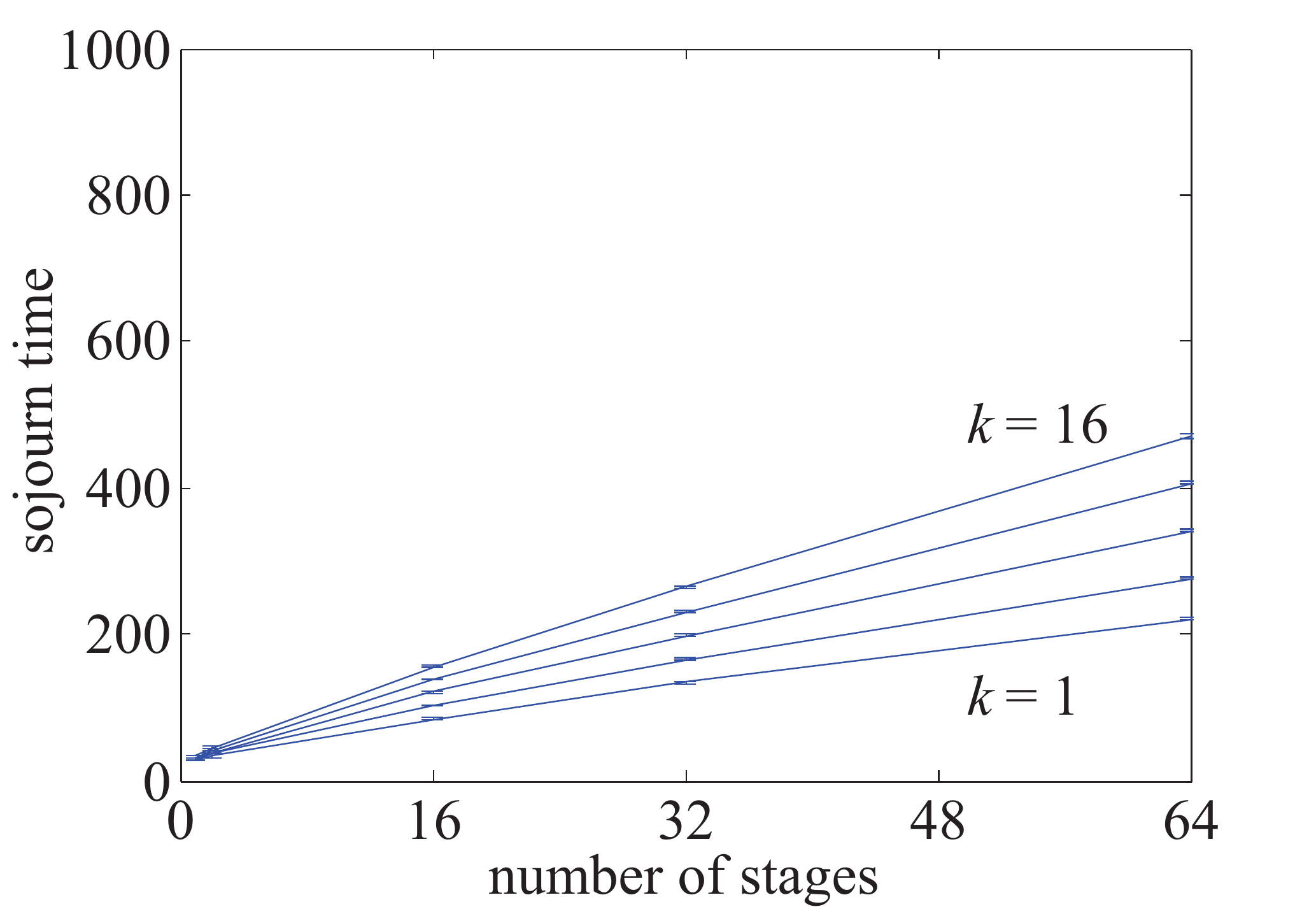}
\label{fig:multistagesimulation}
}
\caption{Multi-stage fork-join network. Sojourn time bounds grow in $\mathcal{O}(h \ln k)$ for $h$ stages each with $k$ parallel servers.}
\label{fig:multistage}
\end{figure}

The simulation results that we depict in Fig.~\ref{fig:multistagesimulation} are obtained for tasks that have an independent service time at each stage. We also conducted simulations of a multi-stage fork-join network where the service times of tasks are identical at each stage, i.e., not independent. While we omit showing the results, it is interesting to note that the end-to-end sojourn times observed in these simulations grow faster than linearly with $h$ as predicted by $\mathcal{O}(h \ln (hk))$ in~\cite{fidler:forkjoin}.
%
%
\section{Multi-Server Systems with Thinning}
\label{sec:thinning}
In this section, we compare the performance of fork-join systems to that of traditional multi-server systems. An example of a traditional multi-server system with $k$ servers is depicted in Fig.~\ref{fig:trqueue}. The difference to fork-join systems is that jobs are not divided into tasks that are served in parallel but instead each job is assigned in its entirety to one of the servers. As a consequence, the external arrival process $A(n)$ is divided into $k$ ``{\it thinned}'' processes $A_i(m)$. The departures $D_i(m)$ of each server $i \in [1,k]$ may optionally be resequenced in the original order of $A(n)$ to form the departure process $D(n)$.
%
%
\subsection{Thinning}
First we introduce some notation. While the external arrival process $A(n)$ specifies the arrival time of job $n$ before thinning, the processes $A_i(m)$ denote the arrival time of the $m$th job of the thinned process at server $i \in [1,k]$. The corresponding service time is $L_i(m)$.

In the case of random thinning, each job is assigned to one of the $k$ servers according to iid discrete (not necessarily uniform) random variables with support $[1,k]$. From the iid property, the mapping of each job to a certain server $i \in [1,k]$ is an independent Bernoulli trial with parameter $p_i$ where $\sum_{i=1}^k p_i = 1$. Let $X_i(m)$ denote the number of the job that becomes the $m$th job that is assigned to server $i$. It follows that $X_i(m)$ is a sum of $m$ iid geometric random variables with parameter $p_i$; i.e., $X_i(m)$ is negative binomial. The arrival process at server $i$ is
\begin{equation}
A_i(m) = A(X_i(m)) ,
\label{eq:randomthinning}
\end{equation}
for $m \ge 1$. Conversely, given jobs $1,2,\dots,n$ of the external arrival process, let $Y_i(n)$ denote the number of jobs assigned to server $i$. It follows that $Y_i(n)$ is binomial with parameter $p_i$. Further, it holds that $X_i(Y_i(n)) \le n$ for $n \ge 1$.

In the case of deterministic thinning a round robin assignment of the jobs of an arrival process $A(n)$ to $k$ servers results in the processes $A_i(m)$ as in \eqref{eq:randomthinning} where
\begin{equation}
X_i(m) = k(m-1)+i ,
\label{eq:splitarrivals}
\end{equation}
for $m \ge 1$ and $i \in [1,k]$. Given jobs $1,2,\dots,n$, the number of jobs that are assigned to server $i$ is
\begin{equation}
Y_i(n) = \left\lceil \frac{n-i+1}{k} \right\rceil ,
\label{eq:mergeddepartures}
\end{equation}
for $n \ge 1$ and $i \in [1,k]$. To see this, note that job $n$ of the external arrival process becomes the $m=\lceil n/k\rceil$th job of server $j = (n-1) \! \mod k + 1$. Hence, $Y_i(n) = m$ for $i \le j$ and $Y_i(n) = m-1$ for $i > j$. The same can be verified for \eqref{eq:mergeddepartures}.
\begin{figure}
  \centering
  \includegraphics[width=0.98\columnwidth]{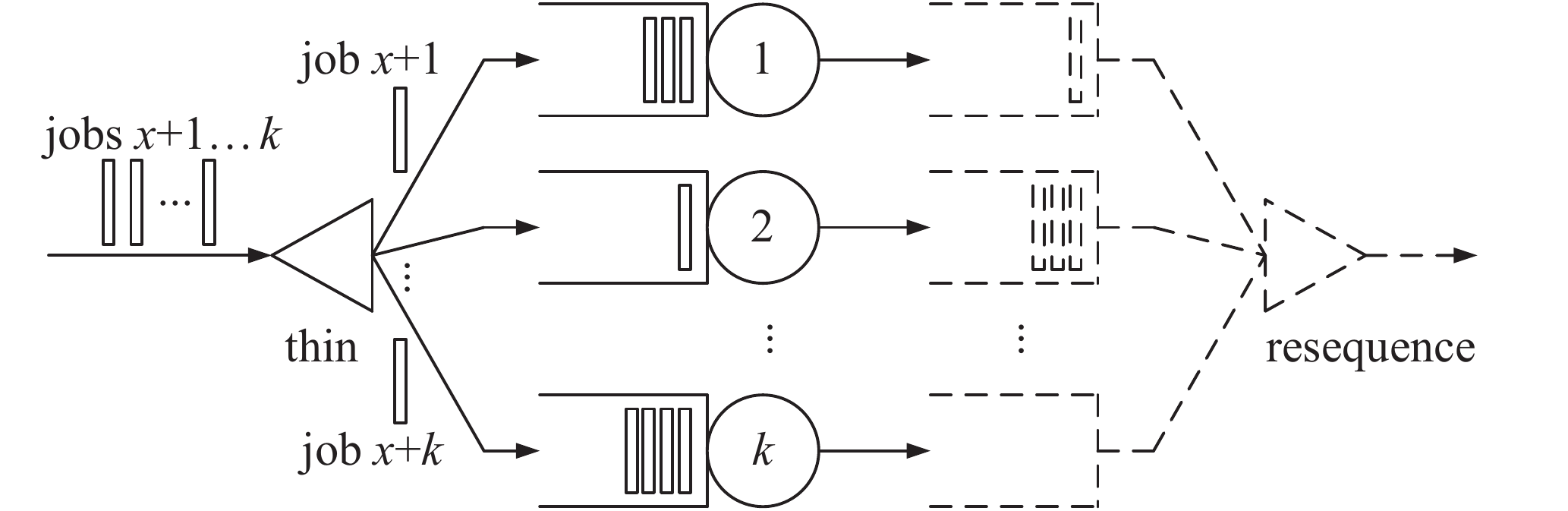}
  \caption{Multi-server system with thinning. Compared to a fork-join system, jobs are not divided into tasks. Instead, entire jobs are assigned to the parallel servers, e.g., deterministically in round robin order, resulting in a thinning of the arrival process. The multi-server system does not maintain the order of jobs unless an additional resequencing step is added (dashed lines).}
  \label{fig:trqueue}
\end{figure}
\begin{corollary}[Thinning]
\label{cor:thinning}
Given arrivals with iid inter-arrival times and parameter $\rho_A(-\theta)$ for $\theta > 0$ as in \eqref{eq:arrivalparameter}. In the case of random thinning with probabilities $p_i$ for $i \in [1,k]$, the thinned arrival processes have parameter
\begin{equation*}
\rho_{A_i}(-\theta) = -\frac{1}{\theta} \ln \biggl( \frac{p_i e^{-\theta\rho_A(-\theta)}}{1-(1-p_i) e^{-\theta\rho_A(-\theta)}} \biggr) ,
\end{equation*}
where $\theta > 0$ so that $e^{-\theta\rho_A(-\theta)} < 1/(1-p_i)$.

In the case of deterministic thinning, for $\theta > 0$ the thinned arrival processes have parameter
\begin{equation*}
\rho_{A_i}(-\theta) = k \rho_A (-\theta) .
\end{equation*}
\end{corollary}
\begin{proof}
The thinned arrivals are expressed by \eqref{eq:randomthinning} as a doubly random process that has increments $A_i(\nu,\nu+1) = A(X_i(\nu),X_i(\nu+1))$ for $\nu \ge 1$.
Considering iid inter-arrival times, the MGF of the thinned process is
\begin{equation}
\mathsf{M}_{A_i(\nu,\nu+1)}(-\theta) = \mathsf{E} \left[(\mathsf{M}_{A(1,2)}(-\theta))^{X_i(1)} \right],
\label{eq:thinningmgf}
\end{equation}
for $\nu \ge 1$. After some reordering, it follows that
\begin{equation}
\mathsf{M}_{A_i(\nu,\nu+1)}(-\theta) = \mathsf{M}_{X_i(1)}(\ln \mathsf{M}_{A(1,2)}(-\theta)) .
\label{eq:doublystochasticarrivalmgf}
\end{equation}
Since $X_i(1)$ is a geometric random variable with MGF $\mathsf{M}_{X_i(1)}(\theta) = p_i e^{\theta}/ (1-(1-p_i)e^{\theta})$ for $\theta < -\ln(1-p_i)$, we obtain by insertion of \eqref{eq:doublystochasticarrivalmgf} into \eqref{eq:arrivalparameter} that
\begin{equation}
\rho_{A_i}(-\theta) = -\frac{1}{\theta} \ln \left( \frac{p_i \mathsf{M}_{A(1,2)}(-\theta)}{1-(1-p_i) \mathsf{M}_{A(1,2)}(-\theta) }\right) ,
\label{eq:splitinterarrivaltimes}
\end{equation}
where $\theta > 0$ so that $\mathsf{M}_{A(1,2)}(-\theta) < 1/(1-p_i)$.

In the case of deterministic thinning, \eqref{eq:thinningmgf} simplifies to $\mathsf{M}_{A_i(\nu,\nu+1)}(-\theta) = (\mathsf{M}_{A(1,2)}(-\theta))^{k}$ so that \eqref{eq:arrivalparameter} evaluates to $\rho_{A_i}(-\theta) = k \rho_A(-\theta)$.
\end{proof}

Since each of the servers of the multi-server system serves its thinned arrivals independently, statistical performance bounds follow straightforwardly by insertion of the arrival parameters from Cor.~\ref{cor:thinning} into Th.~\ref{th:gg1}. The main effect of thinning is captured in the stability condition of Th.~\ref{th:gg1}, which becomes $\rho_{L}(\theta) < \rho_{A_i}(-\theta)$, where $\rho_{A_i}(-\theta)$ increases with $k$. Given fixed $\rho_{L}(\theta)$, the increase of $\rho_{A_i}(-\theta)$ permits larger $\theta$ that result in a faster tail decay. Further, the modularity of this approach implies that the servers may themselves be fork-join systems that can be analyzed by insertion of Cor.~\ref{cor:thinning} into Cor.~\ref{cor:forkjoin}. The results enable a comparison of the performance of fork-join systems with multi-server systems.

First, we investigate how performance bounds for the multi-server system with deterministic thinning grow with $k$. To achieve comparability with the fork-join results, the multi-server system serves jobs that are composed of $k$ tasks. Jobs are, however, served in their entirety by one of the servers. Given iid task service times with parameter $\rho_{Q}(\theta)$ for $i \in [1,k]$ as defined by \eqref{eq:serviceparameter}, the job service times have parameter $\rho_L(\theta) = k \rho_{Q}(\theta)$. By insertion of $\rho_L(\theta)$ and $\rho_{A_i}(-\theta)$ from Cor.~\ref{cor:thinning} into Th.~\ref{th:gg1}, we obtain the stability condition $k \rho_{Q}(\theta) < k \rho_A(-\theta)$. Hence, the maximal $\theta$ that achieves the stability condition is independent of $k$. As a consequence, the waiting time bound from Th.~\ref{th:gg1} and the speed of the tail decay of the sojourn time bound do not depend on $k$. Regarding the sojourn time, we equate the bound from Th.~\ref{th:gg1} with $\varepsilon$ and solve for
\begin{equation}
\tau \le k \rho_{Q}(\theta) + \frac{\ln \alpha - \ln \varepsilon}{\theta} ,
\label{eq:multiserverlingrowth}
\end{equation}
for $\theta > 0$ under the stability condition $\rho_{Q}(\theta) < \rho_A(-\theta)$. Eq.~\eqref{eq:multiserverlingrowth} shows a linear growth with $k$. The result compares to the logarithmic growth established by \eqref{eq:forkjoinlogkgrowth} for the fork-join system. If $k$ is large, the sojourn time of a job at the multi-server system is dominated by its service time that depends linearly on $k$. The fork-join system avoids this effect, as the tasks of the jobs are served by $k$ servers in parallel.

We also evaluate a hybrid system where the arrivals are divided into $a$ thinned processes that are served by $a$ fork-join sub-systems each with $b$ servers. This type of system was also studied as partial mapping in~\cite{rizk:forkjoin}. The overall system comprises $k = ab$ servers. Given jobs that consist of $k$ tasks, each server of the selected fork-join system has to serve $k/b = a$ tasks per job. Following the same steps as above, we obtain by insertion of Cor.~\ref{cor:thinning} into Cor.~\ref{cor:forkjoin}:
\begin{equation}
\tau \le a \rho_{Q}(\theta) + \frac{\ln b + \ln \alpha - \ln \varepsilon}{\theta} ,
\label{eq:combinedsystemgrowth}
\end{equation}
for $\theta > 0$ under the stability condition $a \rho_{Q}(\theta) < a \rho_A(-\theta)$. As suggested by \eqref{eq:combinedsystemgrowth}, configurations where $a \sim \ln k$ achieve a logarithmic scaling with $k$.
\paragraph*{M$\mid$$\text{E}_\text{k}$ jobs}
For numerical evaluation, we consider arrivals with iid exponential inter-arrival times and jobs with iid Erlang-$k$ service times with parameters $\lambda$ and $\mu$, respectively. Each job consists of $k$ tasks with iid exponential service times with parameter $\mu$. We consider three system configurations: i) a multi-server system with thinning, ii) a fork-join system, and iii) a hybrid system, all with $k$ servers.

For the multi-server system with thinning (i), performance bounds are derived from Th.~\ref{th:gg1} using the parameters of the thinned arrival processes as in Cor.~\ref{cor:thinning}. Deterministic thinning results in processes $A_i(m)$ where the inter-arrival times are a sum of $k$ exponential random variables; that is, Erlang-$k$ distributed. It follows by insertion of $\rho_A(-\theta)$ from \eqref{eq:expoarrivalparameter} into Cor.~\ref{cor:thinning} that
\begin{equation}
\rho_{A_i}(-\theta) = -\frac{k}{\theta} \ln \left(\frac{\lambda}{\lambda + \theta}\right) ,
\label{eq:expodetthinning}
\end{equation}
for $\theta > 0$. In the case of random thinning we have by insertion of $\rho_A(-\theta)$ from \eqref{eq:expoarrivalparameter} into Cor.~\ref{cor:thinning} with $p_i = 1/k$ that
\begin{equation}
\rho_{A_i}(-\theta) = - \frac{1}{\theta} \ln \left(\frac{\lambda}{\lambda + k \theta}\right) ,
\label{eq:exporandthinning}
\end{equation}
for $\theta > 0$. In this case, the inter-arrival times of the thinned processes are exponentially distributed with parameter $\lambda/k$. Lastly, the Erlang-$k$ service times of the jobs have parameter
\begin{equation}
\rho_{L}(\theta) = \frac{k}{\theta} \ln \left(\frac{\mu}{\mu - \theta}\right) ,
\label{eq:erlangservice}
\end{equation}
for $\theta \in (0,\mu)$. For deterministic thinning, the maximal $\theta$ that satisfies the stability condition $\rho_L(\theta) \le \rho_{A_i}(-\theta)$ is $\theta = \mu-\lambda$.

Fig.~\ref{fig:thinningcomparison} contrasts the tail decay of deterministic and random thinning for $\mu=1$, $\lambda=0.5$, and $k \in \{4,8,12\}$. The speed of the tail decay does not depend on $k$ for deterministic thinning. In contrast, for random thinning the tail decay becomes slower with increasing $k$. Deterministic thinning generally outperforms random thinning. In the following comparison we include only deterministic thinning.
\begin{figure}
  \centering
  \includegraphics[width=0.75\columnwidth]{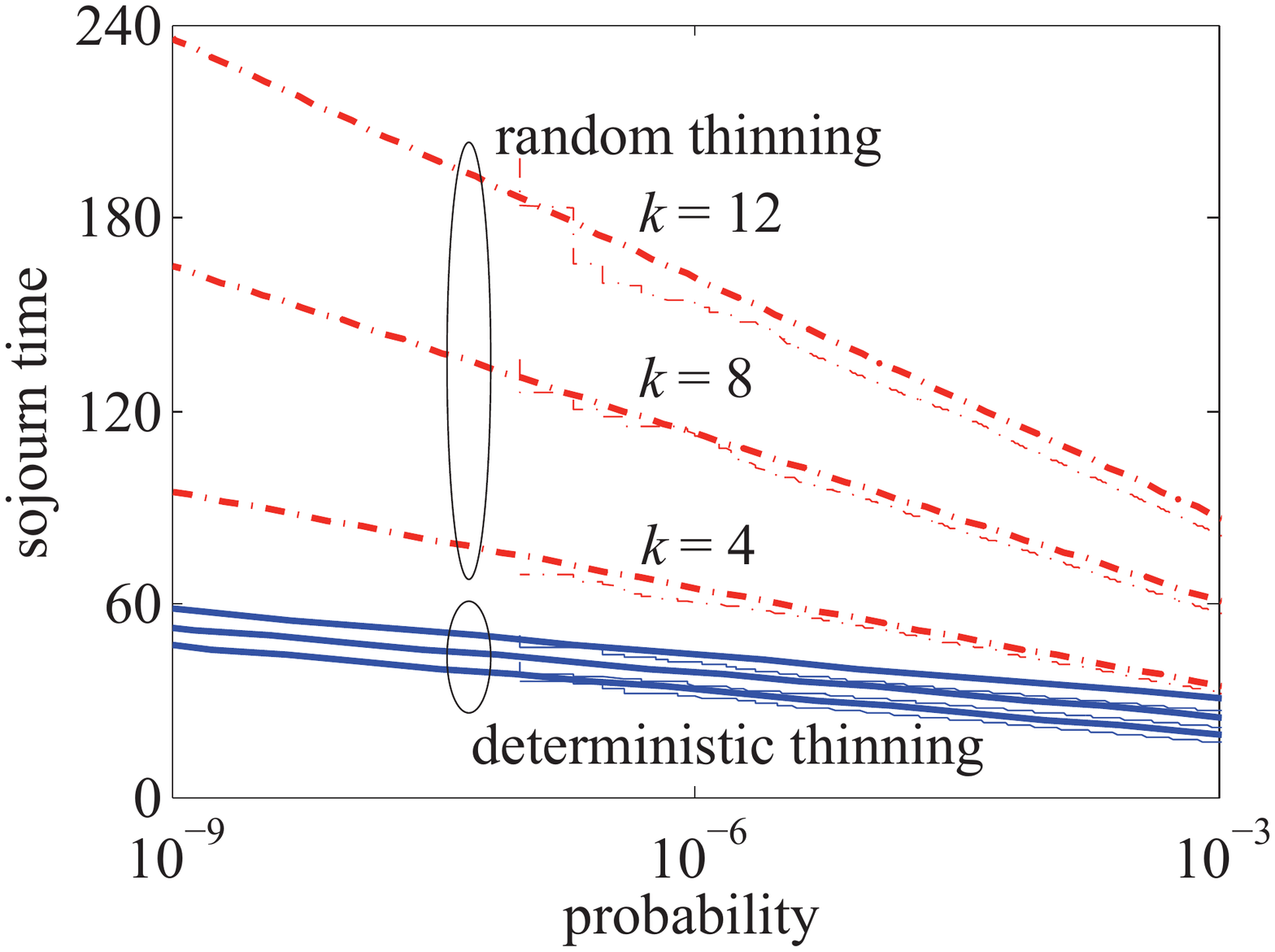}
  \caption{Multi-server system with thinning. Analytical bounds (thick lines) and simulation results (thin lines down to $\varepsilon = 10^{-7}$). Deterministic thinning outperforms random thinning.}
  \label{fig:thinningcomparison}
\end{figure}

The fork-join system (ii) is the same as already evaluated in Fig.~\ref{fig:mm1forkjoin}. For the hybrid system (iii), the thinned arrival process has Erlang-$a$ inter-arrival times and since each job has $k$ tasks that are divided among $b$ servers, each server has to serve $k/b = a$ exponential tasks, resulting in sum in an Erlang-$a$ service time. We choose $a = \log_2 k$ so that $b=k/\log_2 k$. Performance bounds for this system are derived from Cor.~\ref{cor:forkjoin} by insertion of the parameters from Cor.~\ref{cor:thinning}.

\begin{figure*}
  \centering
  \subfigure[thinning vs. fork-join]{
  \includegraphics[width=0.6\columnwidth]{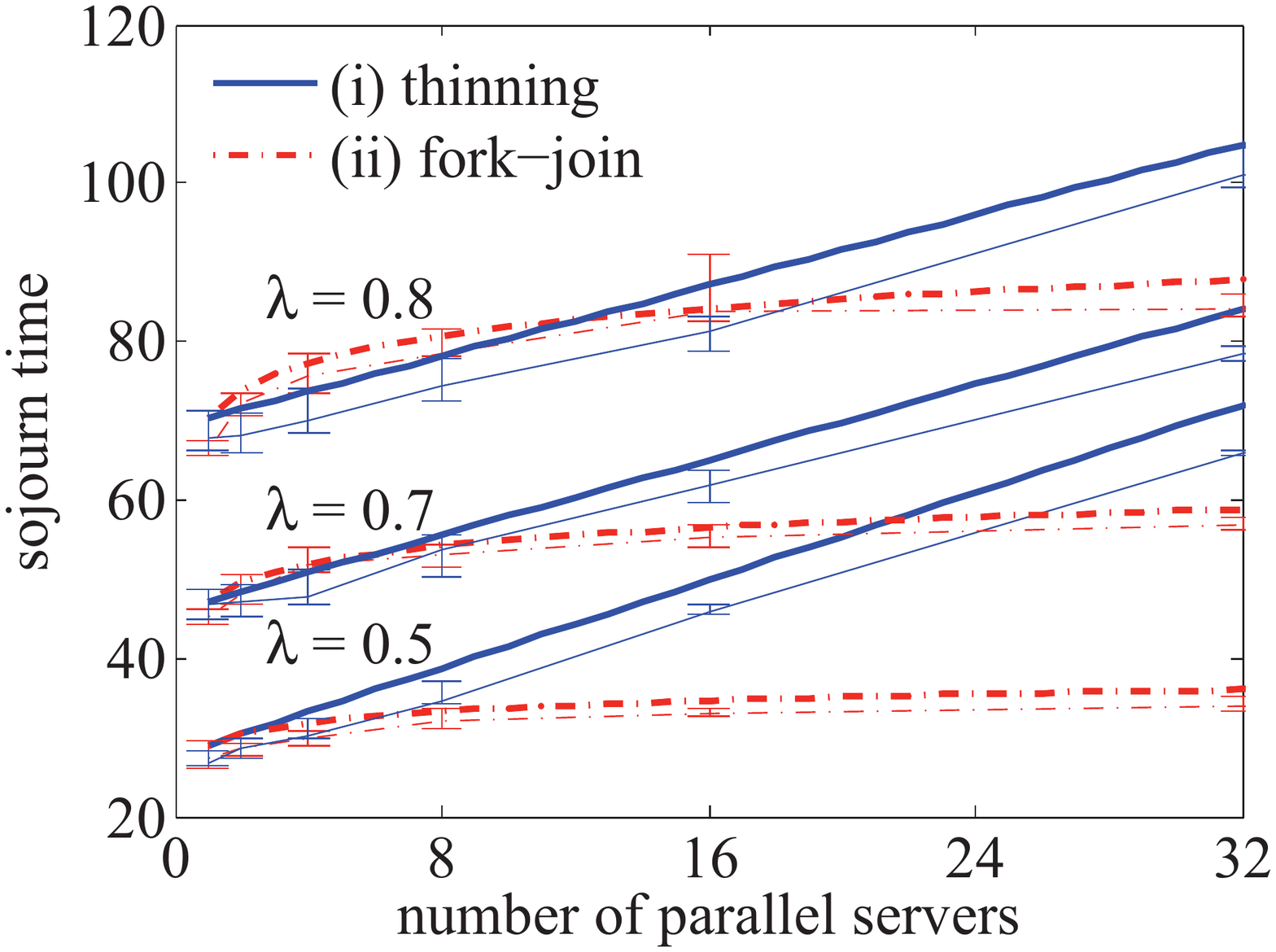}
  \label{fig:thinningforkjoincomparison}
  }
  \hfill
  \subfigure[fork-join vs. hybrid]{
  \includegraphics[width=0.6\columnwidth]{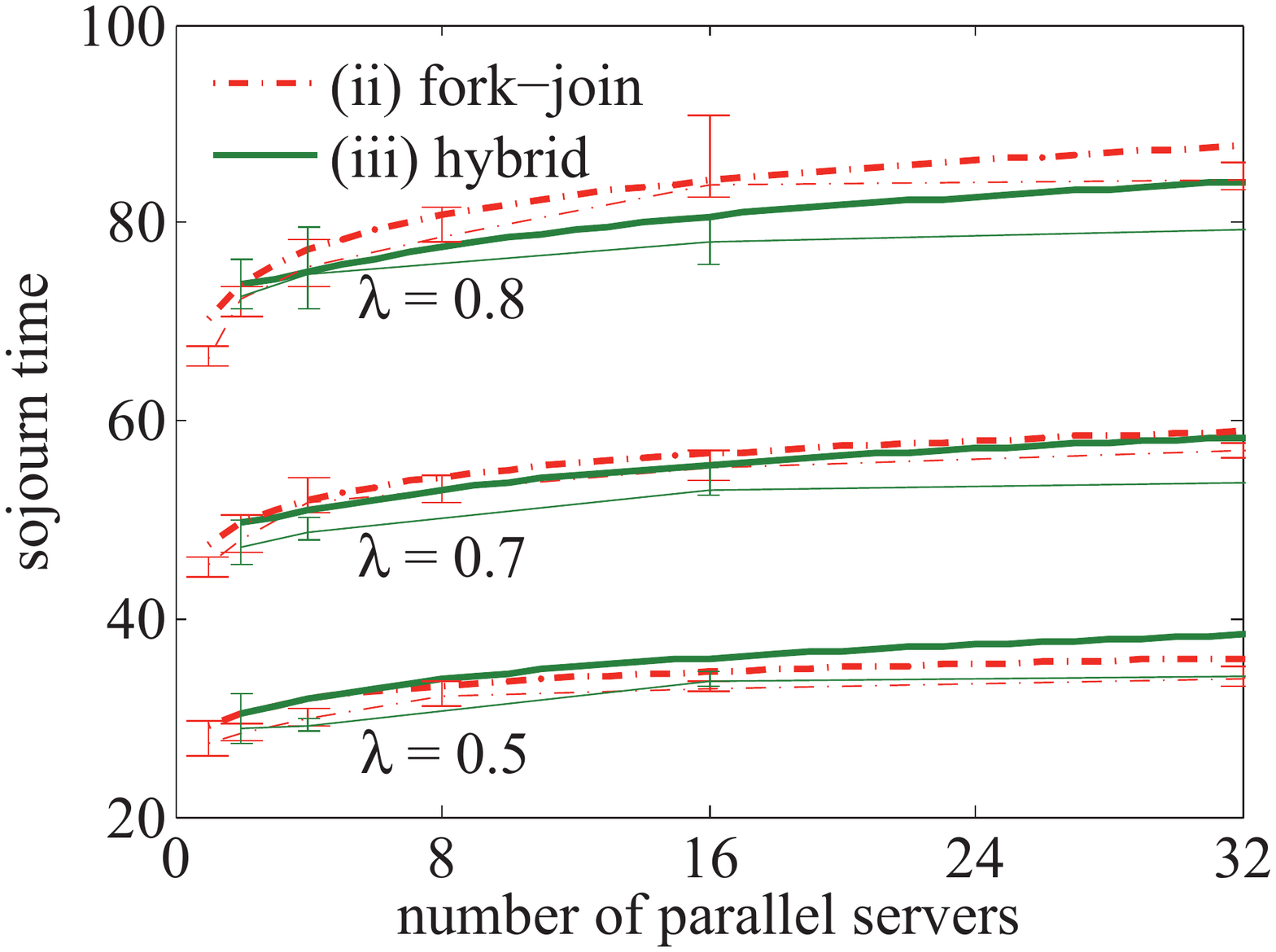}
  \label{fig:forkjoinhybridcomparison}
  }
  \hfill
  \subfigure[thinning with resequencing]{
  \includegraphics[width=0.6\columnwidth]{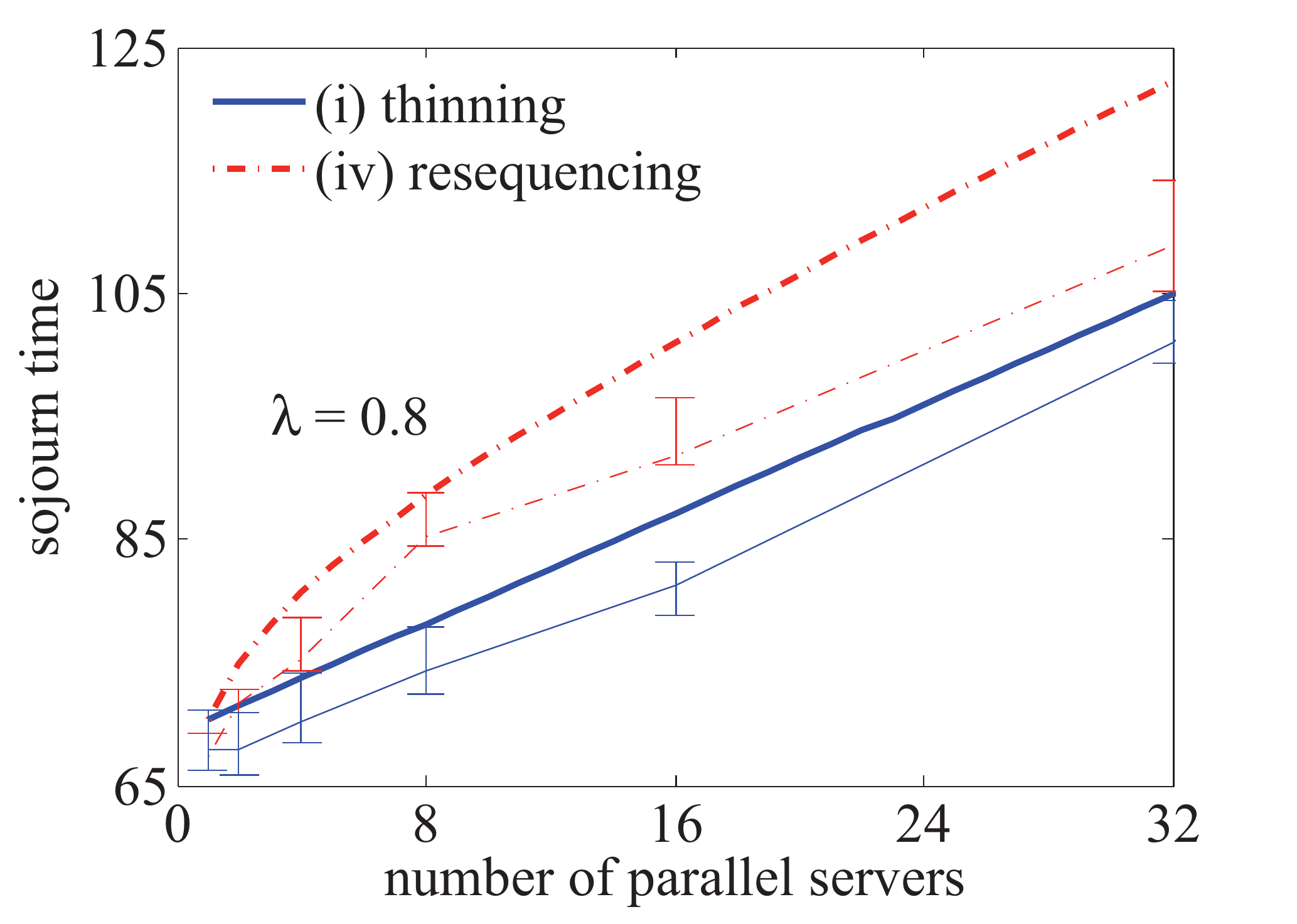}
  \label{fig:thinningresequencingcomparison}
  }
  \caption{Comparison of different multi-server configurations. Analytical bounds (thick lines) and simulation results (thin lines). Sojourn time bounds grow linearly with $k$ in the case of multi-server systems with thinning (i) and with $\ln k$ for fork-join systems (ii), each with $k$ servers. The hybrid system (iii) is optimized to achieve a growth with $\ln k$, as well. Resequencing (iv) adds another delay that grows with $\ln k$ to the multi-server system with thinning (i).}
  \label{fig:multiservercomparison}
\end{figure*}
Fig.~\ref{fig:thinningforkjoincomparison} and Fig.~\ref{fig:forkjoinhybridcomparison} compare the performance of the three systems for $\mu=1$, $\lambda \in \{0.5,0.7,0.8\}$, and $\varepsilon = 10^{-6}$. Clearly, the sojourn time bounds of the multi-server system with deterministic thinning grow linearly with $k$, as established by \eqref{eq:multiserverlingrowth}. For large $k$ the sojourn time of a job is dominated by its service time, so that the curves that are depicted for different arrival rates $\lambda$ converge eventually. The fork-join system mitigates the impact of large jobs by serving the tasks in parallel. It achieves a scaling in $\ln k$, see \eqref{eq:forkjoinlogkgrowth}, that is due to the synchronization constraint of the join operation. The fork-join system mostly outperforms the multi-server system, except in the case of large $\lambda$ and small $k$. The reason is that in a fork-join system the occurrence of a large task blocks all subsequent tasks of that server so the respective jobs cannot complete the join operation. In contrast, in a multi-server system there is no synchronization constraint, so subsequent jobs that are served by other servers can finish service earlier. A similar argument applies for the hybrid system that achieves a logarithmic scaling with $k$, see \eqref{eq:combinedsystemgrowth}, and outperforms the fork-join system if $\lambda$ is large.
%
%
\subsection{Resequencing}
Unlike fork-join systems, multi-server systems do not guarantee that jobs depart in their order of arrival. An optional resequencing step is depicted in Fig.~\ref{fig:trqueue}. As the resequencing step does not affect the waiting time of a job, we state the following corollary only for the sojourn time.
\begin{corollary}[Thinning and resequencing]
\label{cor:thinningresequencing}
Consider a multi-server system with $k$ parallel servers as in Lem.~\ref{lem:exactmaxplusserviceprocess}, thinning, and resequencing. The thinned arrivals have parameter $\rho_{A_i}(-\theta)$ as specified by Cor.~\ref{cor:thinning} and the jobs have service parameters $(\sigma_{L}(\theta),\rho_{L}(\theta))$ as specified by Def.~\ref{def:sigmarho}. For $n \ge 1$, the sojourn time satisfies
\begin{equation*}
\mathsf{P}[ T(n) > \tau ] \le k \alpha e^{\theta \rho_{L}(\theta)} e^{-\theta\tau} .
\end{equation*}
In the case of GI$\mid$G arrival and service processes, the free parameter $\theta > 0$ has to satisfy $\rho_{L}(\theta) < \rho_{A_i}(-\theta)$ and
\begin{equation*}
\alpha = \frac{e^{\theta\sigma_{L}(\theta)}}{1-e^{-\theta (\rho_{A_i}(-\theta)-\rho_{L}(\theta))}} .
\end{equation*}
In the case of GI$\mid$GI arrival and service processes, $\theta > 0$ has to satisfy $\rho_{L}(\theta) \le \rho_{A_i}(-\theta)$ and $\alpha=1$.
\end{corollary}
\begin{proof}
Given the departure processes of each of the servers $D_i(n)$ for $i \in [1,k]$, the combined in-sequence departure process for $n \ge 0$ is
\begin{equation}
D(n) = \max_{i \in [1,k]} \{D_i(Y_i(n))\} ,
\label{eq:randomresequencing}
\end{equation}
where $D_i(0)=0$ by convention. Note that in evaluating \eqref{eq:randomresequencing} one only has to verify the departure of job $Y_i(n)$ from each server $i \in [1,k]$, since the departure of job $Y_i(n)$ from server $i$ implies the departure of all jobs $\nu \in [1,Y_i(n)]$ of the same server. I.e., $D_i(\nu) \le D_i(Y_i(n))$ for $\nu \in [1,Y_i(n)]$, since each server implements first-in first-out order.

Next, we use Lem.~\ref{lem:exactmaxplusserviceprocess} for each server $i \in [1,k]$ to obtain $D_i(Y_i(n)) = \max_{m \in [1,Y_i(n)]} \{A_i(m) + S_i(m,Y_i(n)) \}$. By insertion into \eqref{eq:randomresequencing} we have
\begin{equation}
D(n) = \max_{i \in [1,k]} \left\{ \sup_{m \in [1,Y_i(n)]} \{ A_i(m) + S_i(m,Y_i(n)) \} \right\} ,
\label{eq:resequenceddepartures}
\end{equation}
for $n \ge 1$, where $\sup\{\emptyset\} = 0$. We estimate the sojourn time $T(n) = D(n) - A(n) \le D(n) - A_i(Y_i(n))$ for $n \ge 1$, where we used that $A(n) \ge A_i(Y_i(n))$ for $i \in [1,k]$. By insertion of \eqref{eq:resequenceddepartures}, it follows for $n \ge 1$ that
\begin{equation*}
T(n) \le \max_{i \in [1,k]} \left\{ \sup_{m \in [1,Y_i(n)]} \{S_i(m,Y_i(n)) - A_i(m,Y_i(n)) \} \right\} .
\end{equation*}
By the union bound $\mathsf{P}[T(n) > \tau] \le \sum_{i=1}^k \mathsf{P}[T_i(n) > \tau]$, where $T_i(n) \le \sup_{m \in [1,Y_i(n)]} \{S_i(m,Y_i(n)) - A_i(m,Y_i(n)) \}$. Finally, we estimate $\mathsf{P}[T_i(n) > \tau]$ using Th.~\ref{th:gg1}.
\end{proof}
To evaluate the effect of $k$, we consider the case where the service time grows with $k$, expressed as $\rho_L(\theta) = k \rho_{Q}(\theta)$, as before. We investigate deterministic thinning and equate the sojourn time bound from Cor.~\ref{cor:thinningresequencing} with $\varepsilon$ to solve for
\begin{equation}
\tau \le k \rho_{Q}(\theta) + \frac{\ln k + \ln \alpha - \ln \varepsilon}{\theta}
\label{eq:multiserverreseqgrowth}
\end{equation}
for $\theta > 0$ under the stability condition $k\rho_{Q}(\theta) < k\rho_A(-\theta)$. Eq.~\eqref{eq:multiserverreseqgrowth} shows two effects: a linear growth with $k$ that is due to the increase of the job service time, as also observed for the multi-server system in \eqref{eq:multiserverlingrowth}, and a logarithmic term that is due to resequencing.
%
%
\paragraph*{M$\mid$$\text{E}_\text{k}$ jobs}
Fig.~\ref{fig:thinningresequencingcomparison} shows a numerical comparison of multi-server systems with deterministic thinning, with resequencing (iv) and without resequencing (i). We use the same parameters as above. The results clearly show the additional logarithmic delay due to resequencing. Compared to the fork-join system, resequencing consumes the advantage that multi-server systems with thinning showed for large $\lambda$ in Fig.~\ref{fig:thinningforkjoincomparison}.
%
%
\section{Non-Idling Systems}
\label{sec:nonidling}
A major drawback of the basic fork-join system, as well as the multi-server system with thinning, is that servers may idle while tasks (or jobs, respectively) are queued at other servers. This is due to the early and static assignment of tasks to servers at their time of arrival, that necessitates an individual queue at each server. One advantage of the early assignment is that data can be prefetched at the respective server during the waiting time. If a late assignment of tasks to servers is possible, an implementation with a single queue avoids idling; that is, whenever a server becomes idle, the system assigns the task (or job) at the head of the queue to that server. Compared to multi-queue systems, where each server has an individual queue, single-queue systems use additional feedback information from the servers to perform the late assignment of tasks. We assume that feedback and task assignment are instantaneous. Otherwise, e.g., in a distributed implementation, each task incurs an overhead time in addition to its service time. A further difficulty of single-queue systems is that jobs can depart in an order that differs from the order of their arrival, i.e., $D(n) \ngeqslant D(n-1)$. This implies also that the waiting time of job $n$ cannot simply be determined from $D(n-1)$ as in \eqref{eq:waitingtime} and Th.~\ref{th:gg1}, nor from $D(n-k)$. In the following subsections, we first provide performance bounds for single-queue multi-server systems, i.e., G$\mid$M$\mid$$k$ queues. The bounding method that we derive here, extends to single-queue fork-join systems that we present next.
%
%
\subsection{Single-Queue Multi-Server Systems}
We consider single-queue multi-server systems without a resequencing constraint. An example is shown in Fig.~\ref{fig:nimsqueue}. First, we prove that the system satisfies the definition of max-plus server and based thereon we derive performance bounds.
\begin{figure}
  \centering
  \includegraphics[width=0.67\columnwidth]{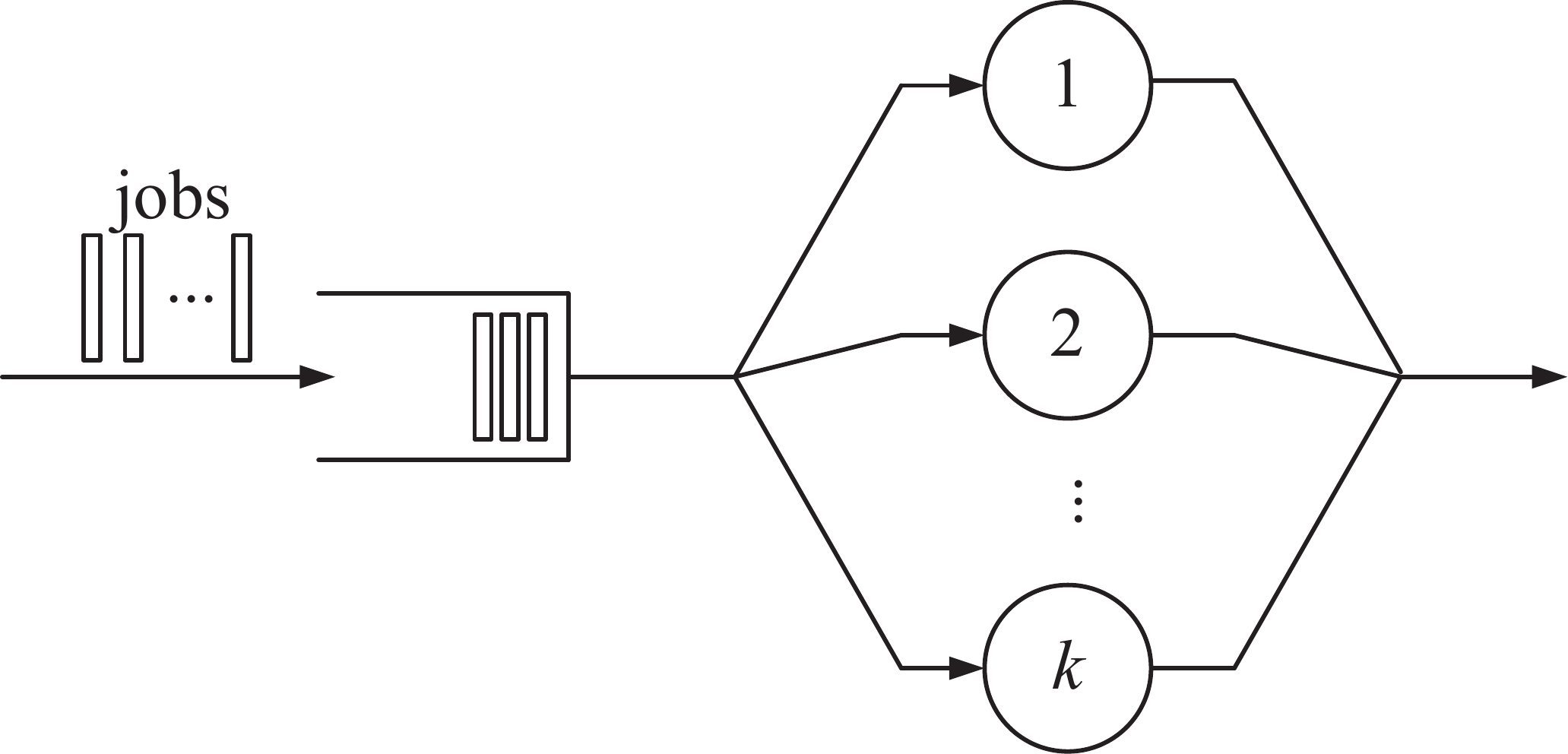}
  \caption{Single-queue multi-server system. The system is non-idling: once a server finishes a job, the head of the line job is assigned to that server.}
  \label{fig:nimsqueue}
\end{figure}
\begin{lemma}[Single-queue multi-server system]
\label{lem:nonidlingmultiserver}
Consider a single-queue multi-server system with $k$ parallel servers as in Lem.~\ref{lem:exactmaxplusserviceprocess}. Let $L(n)$ denote the service time of job $n$ for $n \ge 1$. Given that all servers are busy after job $n$ starts service at $V(n)$, define $Z(n)$ to be the time until the next server becomes idle. Otherwise let $Z(n) = 0$. Define for $n \ge m \ge 1$
\begin{equation*}
S(m,n) = L(n) + \sum_{\nu=m}^{n-1} Z(\nu).
\end{equation*}

i) The system is an exact $S(m,n)$ server.

ii) Given that the jobs have iid exponential service times with parameter $\mu$. The non-zero elements of $Z(n)$ are iid exponential random variables with parameter $k\mu$.

iii) Replace the zero elements of $Z(n)$ by iid exponential random variables with parameter $k\mu$ and compute $S(m,n)$ as above. The system is an $S(m,n)$ server.
\end{lemma}
\begin{proof}
Using the definition of $Z(n)$, it holds for $n \ge 2$ that
\begin{equation}
V(n) = \max \{ A(n), V(n-1) + Z(n-1) \} .
\label{eq:starttimemultiserver}
\end{equation}
Further, $V(1) = A(1)$ and since $Z(n) = 0$ for $n \in [1,k-1]$ we have $V(n) = A(n)$ also for $n \in [2,k]$. By recursive insertion of~\eqref{eq:starttimemultiserver} we obtain for $n \ge 1$ that
\begin{equation}
V(n) = \max_{m \in [1,n]} \left\{ A(m) + \sum_{\nu=m}^{n-1} Z(\nu) \right\} .
\label{eq:starttimemultiserversolved}
\end{equation}
Since $D(n) = V(n) + L(n)$ we have with~\eqref{eq:starttimemultiserversolved} that $D(n) = \max_{m \in [1,n]} \{A(m) + S(m,n)\}$, which proves the first part.

Next, we consider iid exponential service times with parameter $\mu$ and investigate the distribution of $Z(n)$ for $n \ge 1$. Given all servers are busy after the start of service of job $n$ at $V(n)$. This implies that there are another $k-1$ jobs with indices smaller $n$ that are already in service at $V(n)$. Due to the memorylessness of the exponential distribution, the residual service time of each of these jobs as well as the service time of job $n$ are iid exponential random variables with parameter $\mu$. Since the minimum of $k$ iid exponential random variables with parameter $\mu$ is an exponential random variable with parameter $k\mu$, it follows that the time until the next server becomes idle is exponentially distributed with parameter $k \mu$.

For the last part, we use that exponential random variables are non-negative. If we replace all $Z(n)$ that are zero by iid exponential random variables with parameter $k\mu$, \eqref{eq:starttimemultiserversolved} becomes $V(n) \le \max_{m \in [1,n]} \{ A(m) + \sum_{\nu=m}^{n-1} Z(\nu) \}$ for $n \ge 1$, and consequently $D(n) \le \max_{m \in [1,n]} \{A(m) + S(m,n)\}$.
\end{proof}

\begin{figure*}
  \centering
  \subfigure[bound vs. exact result]{
  \includegraphics[width=0.6\columnwidth]{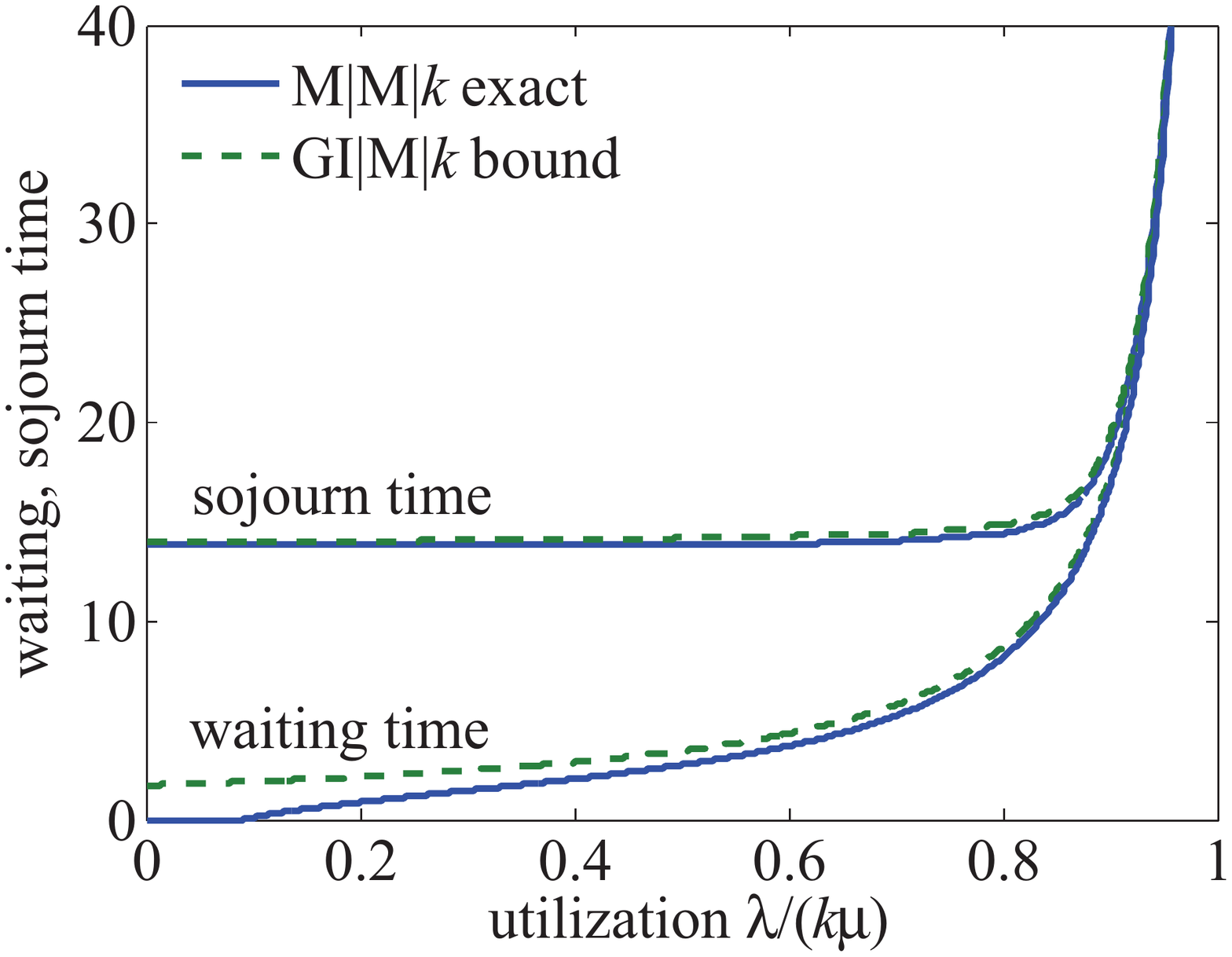}
  \label{fig:nonidlingmultiserverexact}
  }
  \hfill
  \subfigure[tail decay]{
  \includegraphics[width=0.6\columnwidth]{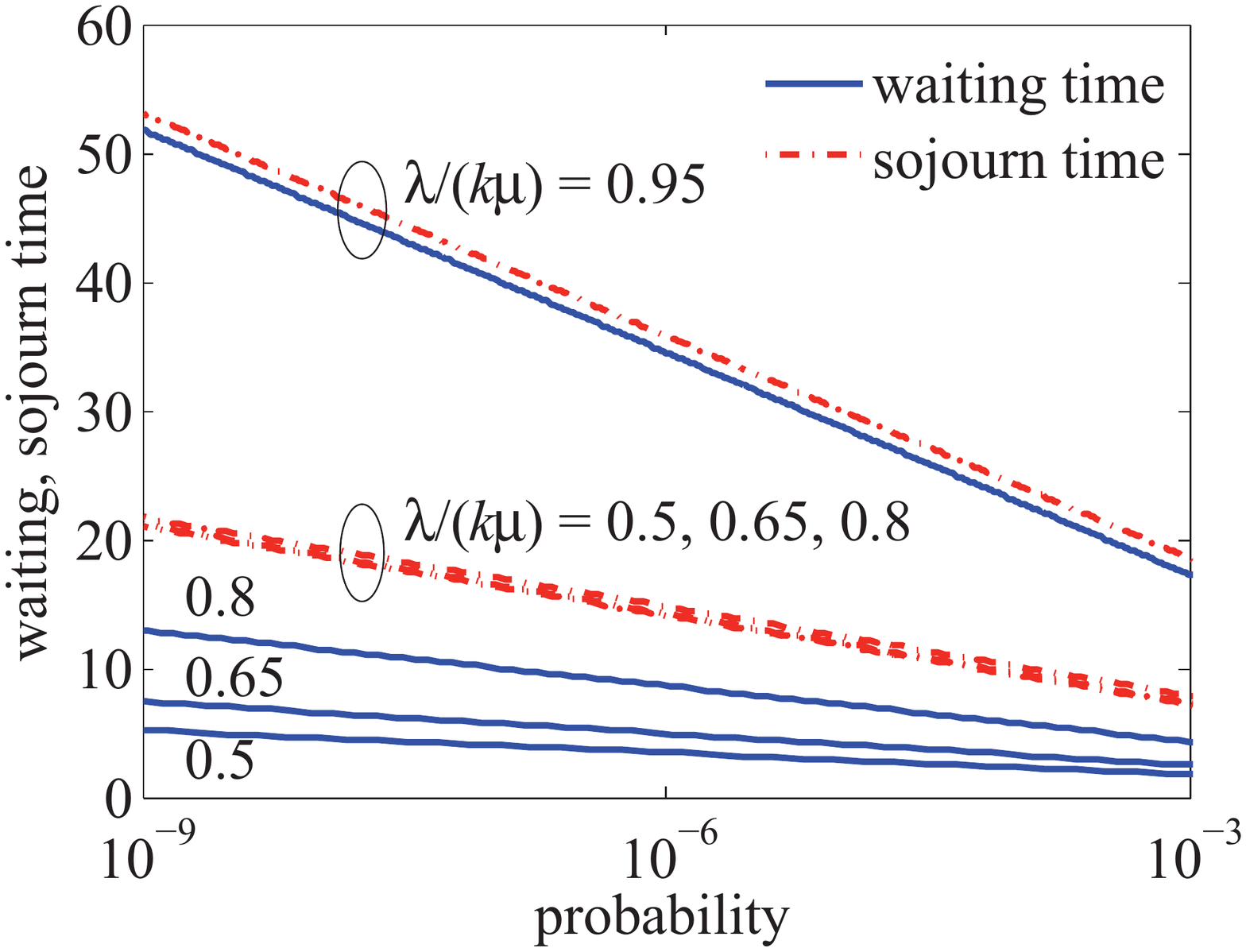}
  \label{fig:nonidlingmultiservertail}
  }
  \hfill
  \subfigure[single- vs. multi-queue system]{
  \includegraphics[width=0.6\columnwidth]{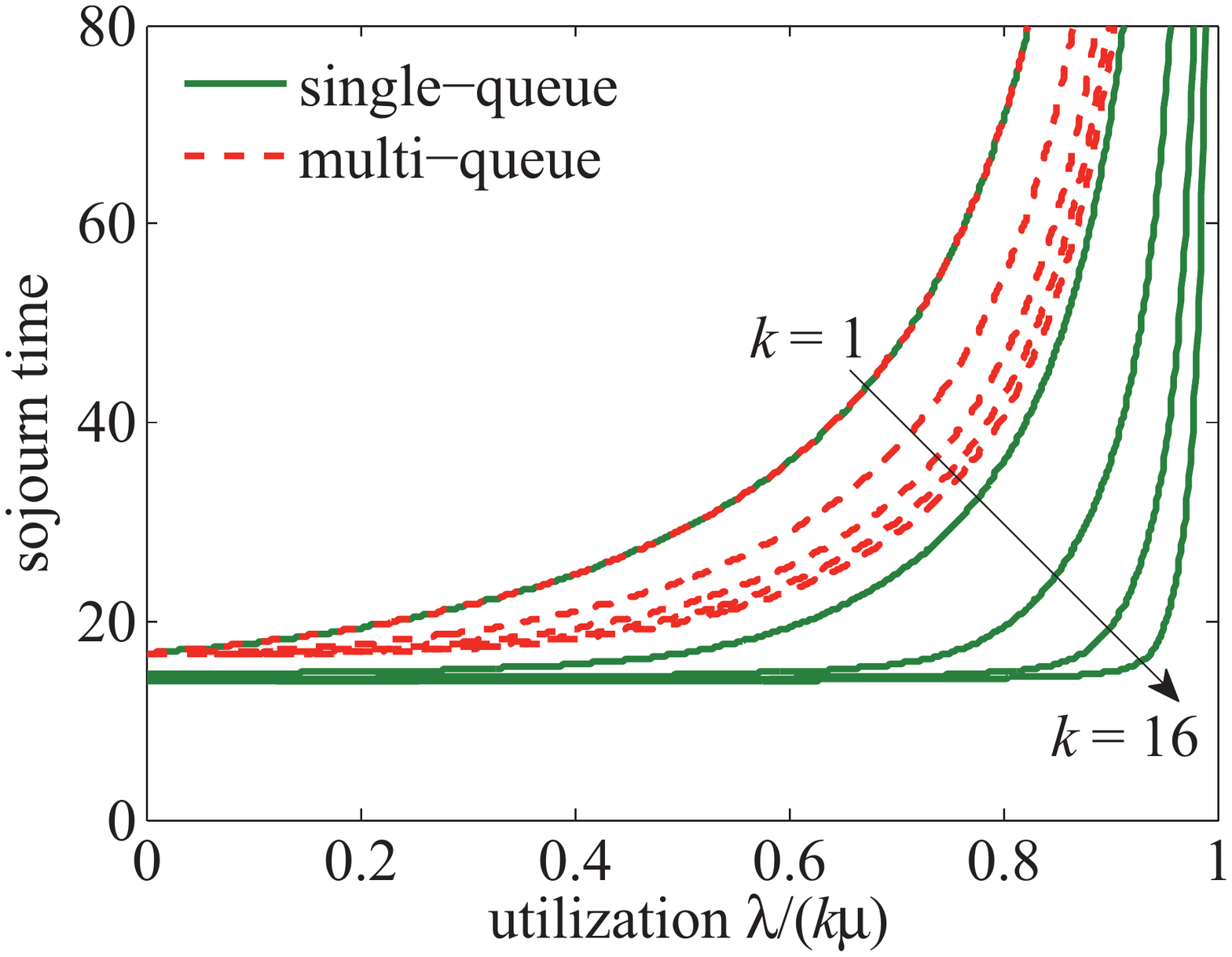}
  \label{fig:nonidlingmultiserverk}
  }
  \caption{Single-queue multi-server system. (a) The bounds agree closely with the exact result. (b) For high utilization the sojourn time is dominated by the waiting time, otherwise by the service time. (c) Comparison of single-queue and multi-queue systems with thinning for $k=\{1,2,4,8,16\}$. The sojourn time decreases with $k$. The improvement is, however, significantly larger in the case of the single-queue system that can sustain a utilization close to one if $k$ is large.}
  \label{fig:nonidlingmultiservercomparison}
\end{figure*}

Considering iid exponential service times $L(n)$ with parameter $\mu$, we have $\rho_L(\theta)$ for $\theta < \mu$ as in \eqref{eq:exposerviceparameter}. With Lem.~\ref{lem:nonidlingmultiserver} (iii), $Z(n)$ is composed of iid exponential random variables with parameter $k\mu$. Invoking \eqref{eq:exposerviceparameter} with parameter $k\mu$ gives
\begin{equation}
\rho_Z(\theta) = \frac{1}{\theta} \ln \left( \frac{k\mu}{k\mu-\theta} \right) ,
\label{eq:rhowaitingnonidling}
\end{equation}
for $\theta < k\mu$. With~\eqref{eq:exposerviceparameter} and~\eqref{eq:rhowaitingnonidling} the MGF of $S(m,n)$ in Lem.~\ref{lem:nonidlingmultiserver} (iii) is $\mathsf{E}\bigl[e^{\theta S(m,n)}\bigr] = e^{\theta (\rho_L(\theta) + \rho_Z(\theta) (n-m))}$ for $\theta < \mu$. Hence, $S(m,n)$ satisfies Def.~\ref{def:sigmarho} with parameters $\sigma_S(\theta) = \rho_L(\theta) - \rho_Z(\theta)$ and $\rho_S(\theta) = \rho_Z(\theta)$.

As verified by Lem.~\ref{lem:nonidlingmultiserver} (iii), the single-queue multi-server system is an $S(m,n)$ server, where $S(m,n)$ is composed of independent increments. Hence, a sojourn time bound can be derived from Th.~\ref{th:gg1} using the parameters $(\sigma_S(\theta),\rho_S(\theta))$ defined above. However, the waiting time bound of Th.~\ref{th:gg1} uses a definition of waiting time $W(n) = [D(n-1) - A(n)]^+$ that does not apply to the single-queue multi-server system, where the departure of job $n-1$ does not generally mark the start of the service of job $n$. The following theorem first formulates the waiting time for single-queue multi-server systems and in a second step derives a sojourn time bound from the waiting time. The derivation avoids the technical limitation $\theta < \mu$ that applies if \eqref{eq:exposerviceparameter} is inserted into Th.~\ref{th:gg1} and thus enables tighter bounds.
\begin{theorem}[Single-queue multi-server system]
\label{th:nonidlingmultiserver}
Consider a single-queue multi-server system as in Lem.~\ref{lem:nonidlingmultiserver}, with arrival parameters $(\sigma_A(-\theta),\rho_A(-\theta))$ as specified by Def.~\ref{def:sigmarho}, iid exponential job service times with parameter $\mu$, and parameter $\rho_Z(\theta)$ given by \eqref{eq:rhowaitingnonidling}. For $n \ge 1$, the sojourn time satisfies
\begin{equation*}
\mathsf{P}[T(n) > \tau] \le \alpha \frac{\mu}{\mu-\theta} \left( e^{-\theta\tau} - e^{-\mu\tau} \right) + e^{-\mu\tau} ,
\end{equation*}
and the waiting time
\begin{equation*}
\mathsf{P}[W(n) > \tau] \le \alpha e^{-\theta\tau} .
\end{equation*}
In the case of G$\mid$M arrival and service processes, the free parameter $0 < \theta < k\mu, \theta \neq \mu$ has to satisfy $\rho_{Z}(\theta) < \rho_A(-\theta)$ and
\begin{equation*}
\alpha = \frac{e^{\theta \sigma_A(-\theta)}}{1-e^{-\theta (\rho_A(-\theta) - \rho_Z(\theta))}} .
\end{equation*}
In the special case of GI$\mid$M arrival and service processes, $0 < \theta < k\mu, \theta \neq \mu$ has to satisfy $\rho_{Z}(\theta) \le \rho_A(-\theta)$ and $\alpha=1$.
\end{theorem}
\begin{proof}
With Lem.~\ref{lem:nonidlingmultiserver} (iii) and \eqref{eq:starttimemultiserversolved}, the waiting time $W(n) = V(n) - A(n)$ for $n \ge 1$ is estimated as
\begin{equation*}
W(n) \le \max_{m \in [1,n]} \left\{ \sum_{\nu=m}^{n-1} Z(\nu) - A(m,n) \right\} ,
\end{equation*}
where $Z(n)$ for $n \ge 1$ are iid exponential random variables with parameter $k\mu$. The derivation of the statistical waiting time bound closely follows the proof of Th.~\ref{th:gg1} in the appendix and is omitted.

To derive a sojourn time bound we use that $D(n) = V(n) + L(n)$ so that $T(n) = D(n) - A(n)$ can be expressed as $T(n) = W(n) + L(n)$, where we substituted $W(n) = V(n) - A(n)$. We use the waiting time bound from Th.~\ref{th:nonidlingmultiserver} to estimate the waiting time CDF\footnote{We note that for $\alpha > 1$, a tighter bound can be derived, using that $F_W(\tau) = 0$ for $\tau < -\ln (1/\alpha)/\theta$. We omit the details for notational brevity.} as $F_W(\tau) = \mathsf{P}[W(n) \le \tau] \ge 1-\alpha e^{-\theta \tau}$. By convolution with the exponential job service time PDF $f_L(\tau) = \mu e^{-\mu \tau}$ for $\tau \ge 0$ we obtain the CDF of the sojourn time as
\begin{equation*}
F_T(\tau) = \int_{0}^{\tau} F_W(\tau-x) f_L(x) dx
\end{equation*}
that evaluates for $\theta \neq \mu$ to
\begin{equation*}
F_T(\tau) \ge 1 - e^{-\mu \tau} - \alpha \frac{\mu}{\mu - \theta} (e^{-\theta \tau} - e^{-\mu \tau}) ,
\end{equation*}
which completes the proof.
\end{proof}
%
%
\paragraph*{M$\mid$M$\mid$$k$ Queue}
For numerical evaluation, we use jobs with exponential inter-arrival times with parameter $\lambda$ and exponential service times with parameter $\mu=1$. In this case, the system becomes the well-known M$\mid$M$\mid$$k$ queue for which there is an exact solution. We note that Th.~\ref{th:nonidlingmultiserver} is not limited to exponential arrivals. Fig.~\ref{fig:nonidlingmultiserverexact} compares the waiting time and the sojourn time bounds with the exact results for $k=8$, $\varepsilon=10^{-6}$, and different utilizations defined as $\lambda/(k\mu)$. The exact waiting time distribution of the M$\mid$M$\mid$$k$ queue is $\mathsf{P}[W(n) > \tau] = P_k e^{- (k\mu-\lambda)\tau}$ where $P_k$ is the probability of waiting, i.e., the probability that $k$ or more jobs are in the system~\cite{gross:queueingtheory}. The bounds from Th.~\ref{th:nonidlingmultiserver} are obtained by insertion of $\rho_A(-\theta)$ from \eqref{eq:expoarrivalparameter}. From the stability condition $\rho_Z(\theta) \le \rho_A(-\theta)$ we find $\theta \le k\mu-\lambda$. By the choice of maximal $\theta$, the waiting time bound $\mathsf{P}[W(n) > \tau] \le e^{- (k\mu-\lambda)\tau}$ exhibits the same exponential speed of decay and differs by the prefactor $P_k$ that approaches one for high utilization. Fig.~\ref{fig:nonidlingmultiserverexact} confirms the good agreement of the waiting time bound and shows visible differences only in the case of low utilization.

Regarding the sojourn time bound in Fig.~\ref{fig:nonidlingmultiserverexact}, we generally have good agreement. Here, two effects can be distinguished. These are expressed by the two parts of the sojourn time bound in Th.~\ref{th:nonidlingmultiserver} that decay as $e^{-\theta\tau}$ and $e^{-\mu\tau}$, respectively. In the case of high utilization, the sojourn time is dominated by the waiting time that decays with $e^{-\theta\tau}$ where $\theta = k\mu-\lambda$. Otherwise, if the utilization satisfies $\lambda/(k\mu) < 1-1/k$ (that is $0.875$ for $k=8$ here), it follows that $\theta > \mu$ so that the waiting time decays quickly and the sojourn time is mostly due to the service time of the job itself that decays slower with $e^{-\mu\tau}$. Fig.~\ref{fig:nonidlingmultiservertail} details the effect, again for $k=8$. For utilizations below $0.875$, the waiting time decays faster than the service time so that the sojourn time changes only marginally if the utilization is increased from $0.5$ to $0.8$. In contrast, once the utilization exceeds $0.875$, the waiting time dominates.

Fig~\ref{fig:nonidlingmultiserverk} evaluates the sojourn time of the single-queue multi-server system for $k = \{1,2,4,8,16\}$ with respect to the multi-queue system with thinning. For comparability, we use the same technique to derive the sojourn time from the waiting time as in Th.~\ref{th:nonidlingmultiserver} for both systems. While for $k=1$ the systems are identical, the single-queue multi-server system outperforms the multi-queue system for larger $k$. Given a target sojourn time bound, the single-queue system can sustain a significantly higher utilization.
\paragraph*{G$\mid$D$\mid$$k$ Queue}
While single-queue multi-server systems show a significant advantage if the service times are exponential, we note that this is not generally the case. An example are deterministic service times, i.e., $L(n)=L$ for $n \ge 1$. In this case, the single-queue multi-server system is governed by $V(n) = \max\{A(n) , V(n-k)+L\}$ for $n > k$ and $V(n) = A(n)$ for $n \in [1,k]$. This is an immediate consequence of the deterministic service times, which ensure that all jobs finish service in the order of their arrival. Considering the multi-queue multi-server system with deterministic thinning, job $n$ is assigned to server $i = (n-1) \mod k + 1$ following job $n-k$. Hence job $n$ starts service at $V(n) = \max\{A(n) , V(n-k)+L\}$ for $n > k$ and $V(n) = A(n)$ for $n \in [1,k]$. This shows that the single-queue und the multi-queue multi-server system perform identical in case of deterministic service times.
%
%
\subsection{Single-Queue Fork-Join Systems}
\label{sec:forkjoinnonidling}
In a single-queue fork-join system, jobs are composed of $k$ tasks that are stored in a single-queue. Once any of the $k$ parallel servers becomes idle, it fetches the next task from the head of the queue. An example is shown in Fig.~\ref{fig:sparkqueue}. As tasks may finish service out-of-sequence, the join operation uses a buffer with random access to complete a job immediately once all of its tasks are finished, regardless of the order of arrival. This implies that $D(n) \ngeqslant D(n-1)$. Our analysis of single-queue fork-join systems follows the same essential steps as in the case of the single-queue multi-server systems with the additional resynchronization constraint of the join operation.
\begin{figure}
  \centering
  \includegraphics[width=0.95\columnwidth]{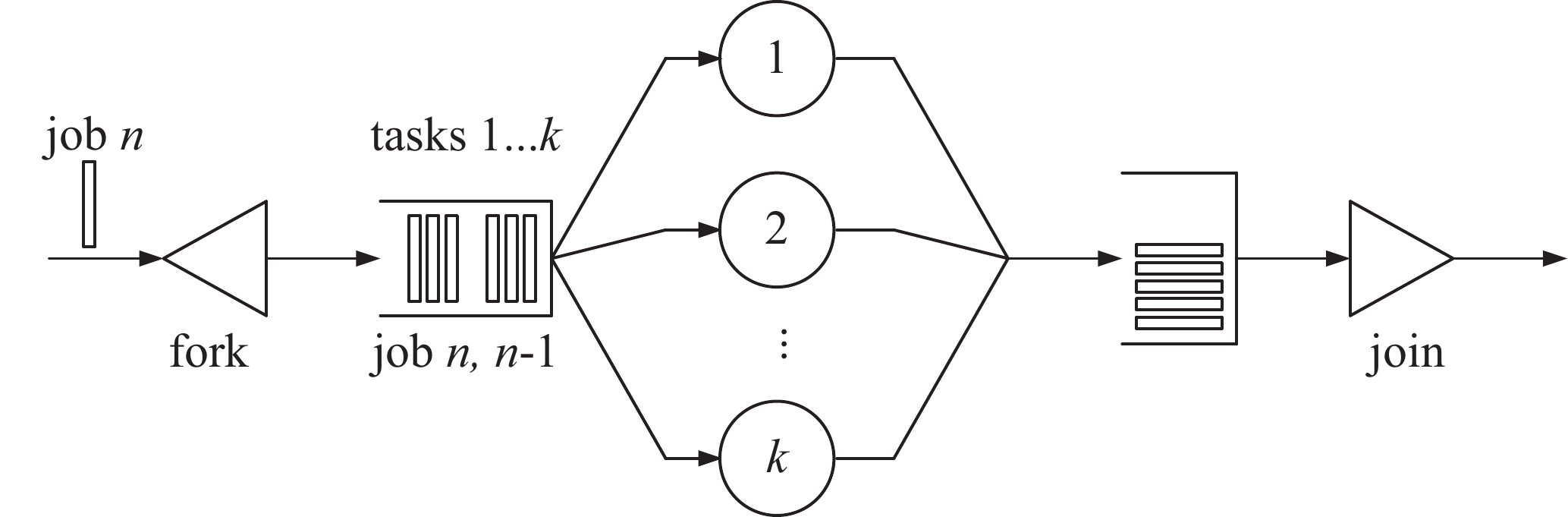}
  \caption{Single-queue fork-join system. The system is non-idling: once a server finishes a task, the head of the line task is assigned to that server. The join operation uses a random access buffer to complete jobs (possibly out of sequence) once all tasks are finished.}
  \label{fig:sparkqueue}
\end{figure}
\begin{lemma}[Single-queue fork-join system]
\label{lem:nonidlingforkjoin}
Consider a single-queue fork-join system with $k$ parallel servers as in Lem.~\ref{lem:exactmaxplusserviceprocess}. Let $Q_i(n)$ denote the service time of task $i$ of job $n$ for $n \ge 1$ and $i \in [1,k]$. Given that all servers are busy after task $i$ of job $n$ starts service at $V_i(n)$, define $Z_i(n)$ to be the time until the next server becomes idle. Otherwise let $Z_i(n) = 0$. Define for $n \ge m \ge 1$
\begin{equation*}
S(m,n) = \max_{i \in [1,k]} \Biggl\{ Q_i(n) + \sum_{j=1}^{i-1} Z_j(n) + \sum_{\nu = m}^{n-1}\sum_{j=1}^{k} Z_j(\nu) \Biggr\} .
\end{equation*}

i) The system is an exact $S(m,n)$ server.

ii) Given that the tasks have iid exponential service times with parameter $\mu$. The non-zero elements of $Z_i(n)$ are iid exponential random variables with parameter $k\mu$.

iii) Replace the zero elements of $Z_i(n)$ by iid exponential random variables with parameter $k\mu$ and compute $S(m,n)$ as above. The system is an $S(m,n)$ server.
\end{lemma}
\begin{proof}
Using the definition of $Z_i(n)$, it holds for $n \ge 1$ and $i \in [2,k]$ that
\begin{equation}
V_i(n) = \max \{ A(n), V_{i-1}(n) + Z_{i-1}(n) \} ,
\label{eq:starttimeforkjoinpart1}
\end{equation}
and for $n \ge 2$ and $i=1$
\begin{equation}
V_1(n) = \max \{ A(n), V_{k}(n-1) + Z_{k}(n-1) \} .
\label{eq:starttimeforkjoinpart2}
\end{equation}
Further, $V_1(1) = A(1)$ and since $Z_i(1) = 0$ for \mbox{$i \in [1,k-1]$} we have $V_i(1) = A(1)$ also for $i \in [2,k]$. By recursive insertion of~\eqref{eq:starttimeforkjoinpart1} and~\eqref{eq:starttimeforkjoinpart2} we obtain for $n \ge 1$ and $i \in [1,k]$ that
%
%
\begin{equation}
V_i(n) = \max_{m \in [1,n]} \Biggl\{A(m) + \sum_{j=1}^{i-1} Z_j(n) + \sum_{\nu = m}^{n-1}\sum_{j=1}^{k} Z_j(\nu)\Biggr\} .
\label{eq:starttimeforkjoinsolved}
\end{equation}
Above we used that $Z_i(n)$ is non-negative to reduce the number of terms that are evaluated by the maximum operator. Given $V_i(n)$, task $i$ of job $n$ finishes service after another $Q_i(n)$ units of time at $D_i(n) = V_i(n) + Q_i(n)$. Finally, job $n$ is completed once all of its tasks have finished service at $D(n) = \max_{i \in [1,k]} \{ V_i(n) + Q_i(n) \}$. Inserting~\eqref{eq:starttimeforkjoinsolved} and reordering the maxima proves the first part.

The proof of the remaining parts is a notational extension of the proof of Lem.~\ref{lem:nonidlingmultiserver} that considers tasks instead of jobs.
\end{proof}

\begin{figure*}
  \centering
  \subfigure[$k=\{1,2,4,8,16\}$, $\varepsilon=10^{-6}$]{
  \includegraphics[width=0.6\columnwidth]{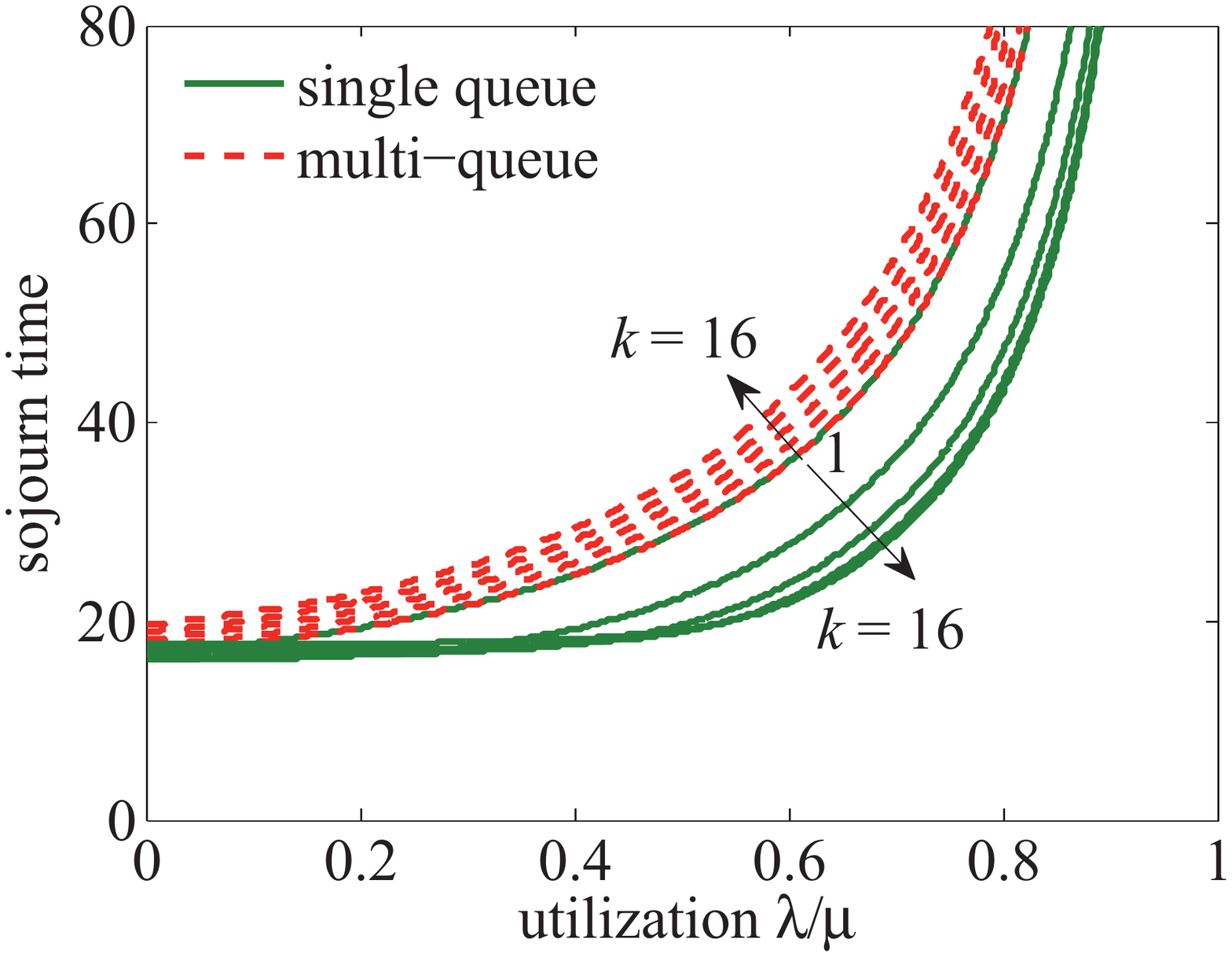}
  \label{fig:nonidlingforkjoink}
  }
  \hfill
  \subfigure[$\lambda=\{0.3,0.7\}$, $\varepsilon=10^{-6}$]{
  \includegraphics[width=0.6\columnwidth]{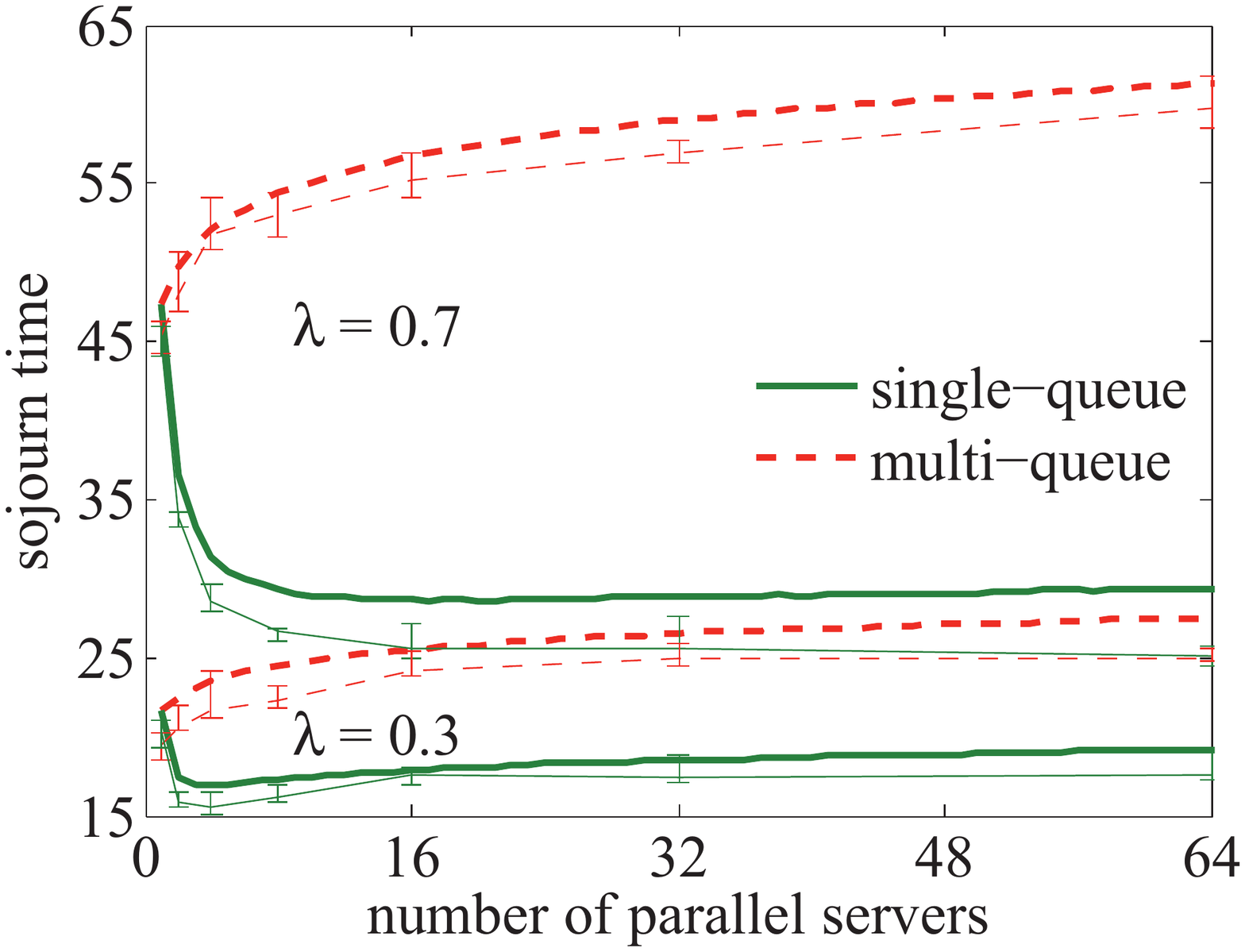}
  \label{fig:nonidlingforkjoink2}
  }
  \hfill
  \subfigure[Spark, $\lambda = 0.7$, $\varepsilon = 10^{-3}$]{
  \includegraphics[width=0.6\columnwidth]{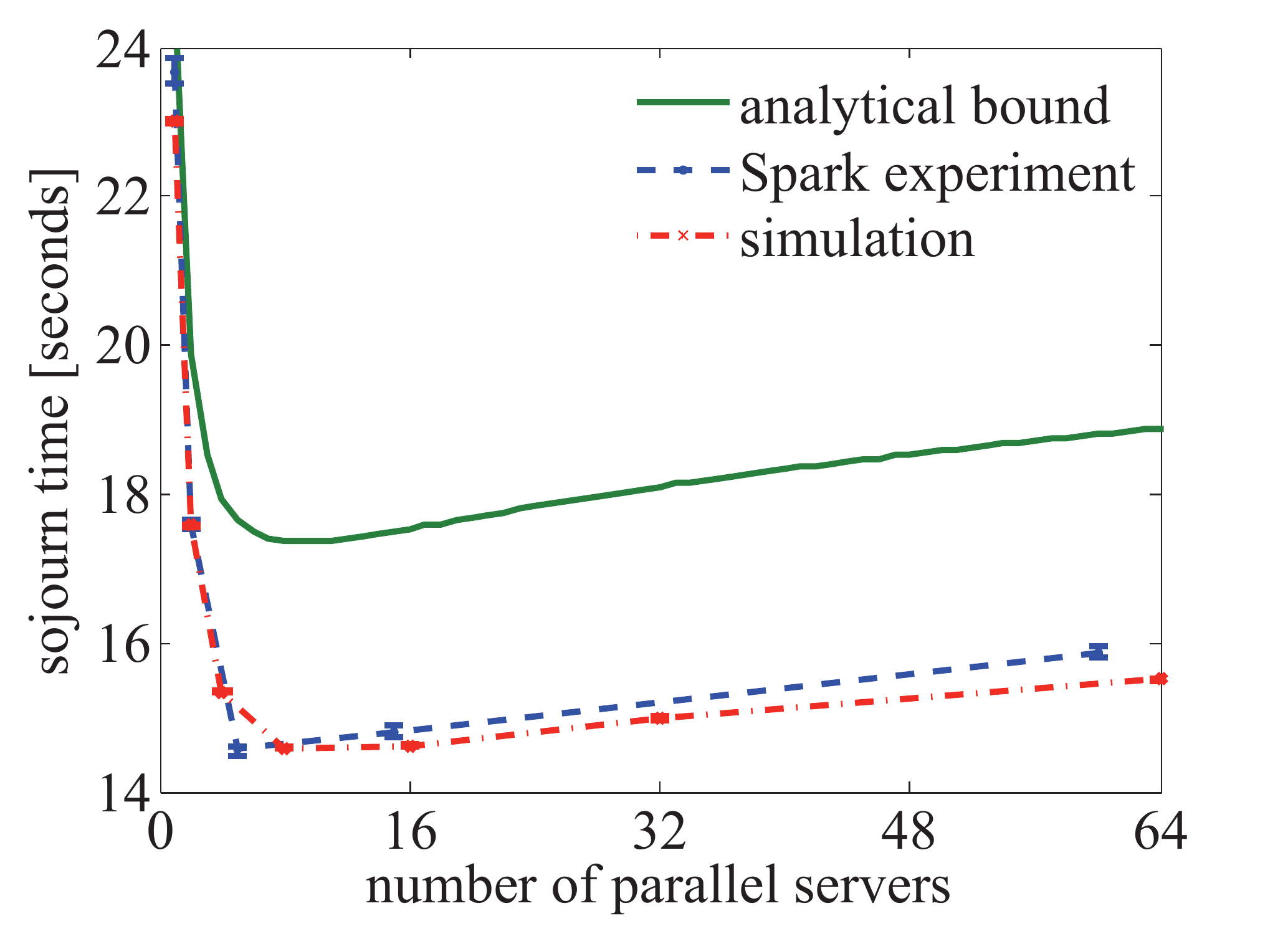}
  \label{fig:spark}
  }
  \caption{(a) and (b) Comparison of single-queue and multi-queue fork-join systems. Analytical bounds (thick lines) and simulation results (thin lines). While the sojourn time of the multi-queue system grows with $\ln k$, the single-queue system achieves a significant improvement due to load-balancing if $k$ is increased, particularly in the case of high utilization. The additional gain diminishes for large $k$ and is eventually consumed by the synchronization constraint of the join operation (parameter $\beta$ in Th.~\ref{th:nonidlingforkjoin}). (c) Spark experiment. The bound predicts the trend of the sojourn time with $k$ of the Spark system.}
  \label{fig:nonidlingforkjoincomparison}
\end{figure*}

Considering exponential service times $Q_i(n)$ with parameter $\mu$, we have $\rho_Q(\theta)$ as in \eqref{eq:exposerviceparameter} for $\theta < \mu$. With Lem.~\ref{lem:nonidlingforkjoin} (iii), $Z_i(n)$ is composed of iid exponential random variables with parameter $k\mu$ that are characterized by $\rho_Z(\theta)$ given by \eqref{eq:rhowaitingnonidling} for $\theta < k\mu$. The MGF of $S(m,n)$ in Lem.~\ref{lem:nonidlingforkjoin} (iii) for $\theta < \mu$ is
\begin{align*}
\mathsf{E}\Bigl[e^{\theta S(m,n)}\Bigr] \le & \sum_{i=1}^k e^{\theta (\rho_{Q}(\theta) + \rho_{Z}(\theta) ((n-m)k+i-1))} \\
& = \beta e^{\theta (\rho_{Q}(\theta) + k \rho_{Z}(\theta) (n-m))} ,
\end{align*}
where
\begin{equation}
\beta = \frac{e^{\theta k \rho_{Z}(\theta)}-1}{e^{\theta \rho_{Z}(\theta)}-1} .
\label{eq:betanonidlingforkjoin}
\end{equation}
Hence, $S(m,n)$ satisfies Def.~\ref{def:sigmarho} with parameters $\sigma_S(\theta) = \rho_{Q}(\theta) - k\rho_{Z}(\theta) + \ln \beta/\theta$ and $\rho_S(\theta) = k\rho_{Z}(\theta)$.

\begin{theorem}[Single-queue fork-join system]
\label{th:nonidlingforkjoin}
Consider a single-queue fork-join system as in Lem.~\ref{lem:nonidlingforkjoin}, with arrival parameters $(\sigma_A(-\theta),\rho_A(-\theta))$ as specified by Def.~\ref{def:sigmarho}, iid exponential task service times with parameter $\mu$, parameter $\rho_{Z}(\theta)$ as specified by \eqref{eq:rhowaitingnonidling}, and $\beta$ as given in \eqref{eq:betanonidlingforkjoin}. For $n \ge 1$, the sojourn time satisfies
\begin{equation*}
\mathsf{P}[T(n) > \tau] \le \alpha \beta \frac{\mu}{\mu-\theta} \left(e^{-\theta \tau}-e^{-\mu \tau}\right) + e^{-\mu \tau},
\end{equation*}
and the waiting time of task $i$
\begin{equation*}
\mathsf{P}[W_i(n) > \tau] \le \alpha e^{\theta (i-1)\rho_{Z}(\theta)} e^{-\theta\tau} .
\end{equation*}
In the case of G$\mid$M arrival and service processes, the free parameter $0 < \theta < k\mu, \theta \neq \mu$ has to satisfy $k\rho_{Z}(\theta) < \rho_A(-\theta)$ and
\begin{equation*}
\alpha = \frac{e^{\theta \sigma_A(-\theta)}}{1-e^{-\theta (\rho_A(-\theta) - k\rho_{Z}(\theta))}} .
\end{equation*}
In the special case of GI$\mid$M arrival and service processes, $0 < \theta < k\mu, \theta \neq \mu$ has to satisfy $k\rho_{Z}(\theta) \le \rho_A(-\theta)$ and $\alpha=1$.
\end{theorem}
We note that the waiting time of the task of job $n$ that starts service last is simply the waiting time of task $k$ of job $n$ as all other tasks start service before, i.e., no maximum as in the multi-queue fork-join system is needed.
\begin{proof}
Similar to the proof of Th.~\ref{th:nonidlingmultiserver}, we start with the waiting time that is expressed as $W_i(n) = V_i(n) - A(n)$ for task $i \in [1,k]$ of job $n \ge 1$. With Lem.~\ref{lem:nonidlingforkjoin} (iii) and \eqref{eq:starttimeforkjoinsolved} we have for $n \ge 1$ that
\begin{equation*}
W_i(n) \le \max_{m \in [1,n]} \Biggl\{\sum_{j=1}^{i-1} Z_j(n) + \sum_{\nu = m}^{n-1}\sum_{j=1}^{k} Z_j(\nu) - A(m,n) \Biggr\} ,
\end{equation*}
where $Z_i(n)$ for $n \ge 1$ and $i \in [1,k]$ are iid exponential random variables with parameter $k\mu$. The derivation of the statistical waiting time bound closely follows the proof of Th.~\ref{th:gg1} in the appendix and is therefore omitted.

A sojourn time bound for each task $i \in [1,k]$ of job $n$ follows from its waiting time bound by convolution with the exponential task service time PDF as in the proof of Th.~\ref{th:nonidlingmultiserver}. Finally, estimating the maximum sojourn time of all tasks $i \in [1,k]$ of job $n$ by use of the union bound leads to parameter $\beta$ defined by~\eqref{eq:betanonidlingforkjoin}.
\end{proof}
\paragraph*{M$\mid$M tasks}
We compare the single-queue fork-join system with the multi-queue system from Sec.~\ref{sec:forkjoinidling}. Jobs have iid exponential inter-arrival times and are composed of $k$ tasks with iid exponential service times. The parameters $\rho_A(-\theta)$ and $\rho_{Q}(\theta)$ are as specified by \eqref{eq:expoarrivalparameter} and \eqref{eq:exposerviceparameter}, respectively, where we let $\mu=1$. For comparability, we use the same technique to derive the sojourn time from the waiting time as in Th.~\ref{th:nonidlingforkjoin} for both systems. In Fig.~\ref{fig:nonidlingforkjoink}, we fix $\varepsilon = 10^{-6}$ and show the impact of the utilization $\lambda/\mu$ for different $k \in \{1,2,4,8,16\}$. For $k=1$ the single-queue and the multi-queue system are identical and the sojourn time bounds from Cor.~\ref{cor:forkjoin} and Th.~\ref{th:nonidlingforkjoin} agree. For increasing $k$, the sojourn time bound of the multi-queue fork-join system shows logarithmic growth with $k$; i.e., the lines are equally spaced. This effect is due to the synchronization constraint of the join operation. In contrast, the sojourn time bounds of the single-queue fork-join system improve with $k$ with decreasing gain. Here, two opposing effects are superimposed: 1.) a gain due to load balancing is achieved by the single-queue system if $k$ is increased, most visible for medium to high utilization; 2.) the synchronization constraint of the join operation, similar to the case of the multi-queue fork-join system. Fig.~\ref{fig:nonidlingforkjoink2} depicts the effects for fixed $\lambda = 0.3$ and $\lambda = 0.7$, respectively, $\varepsilon = 10^{-6}$, and varying $k$. For small $\lambda$ the gain of the single-queue system that is achieved by load-balancing is small. For intuition, if all servers of the two systems are idle at the time of a job arrival, the single-queue and the multi-queue system perform identically and the sojourn time is determined by the task with the maximal service time. For large $\lambda$ the gain of load-balancing becomes more significant as $k$ is increased, but for large $k$ the synchronization constraint of the join operation eventually consumes the gain.

The default scheduler of Apache Spark is a prominent implementation of a single-queue system. Fig.~\ref{fig:spark} shows results for $\lambda = 0.7$, $\mu=1$, and $\varepsilon = 10^{-3}$ from experiments on a live Spark cluster. The units are in seconds. For further details on the experiment see the appendix. Our simulation results match the Spark measurements almost perfectly. Similarly, the sojourn time bound from Th.~\ref{th:nonidlingforkjoin} clearly shows the trend that is observed for Spark if $k$ is increased. First the gain due to load-balancing dominates, and that is later consumed by the synchronization constraint.
%
%
\section{Conclusions}
\label{sec:conclusions}
We formulated a general model of fork-join systems in max-plus system theory and derived performance bounds for fork-join networks with $h$ independent fork-join stages each with $k$ parallel G$\mid$G$\mid$1 servers. The bounds were shown to scale in $\mathcal{O}(h \ln k)$ and compare to the previous result $\mathcal{O}(h \ln (hk))$ that does not take advantage of independence. We performed a detailed comparison of essential configurations of multi-server systems. We included an analysis of single-queue multi-server as well as single-queue fork-join systems that are non-idling as opposed to corresponding multi-queue implementations. We found that the single-queue systems achieve a fundamental performance gain that is due to load-balancing and possible overtaking of jobs. Since jobs can depart out of sequence, the multi-stage analysis of single-queue systems is more difficult and remains as an open research question. We included reference results, mostly from simulation, as well as measurements obtained from Spark experiments which show that the analytical bounds closely predict the actual performance of systems.
%
%
\section*{Appendix A: Simulation and Experiments}
\label{sec:appendixa}
\paragraph*{Forkulator}\texttt{Forkulator} is an event-driven simulator written in Java\footnote{Software available at \url{https://github.com/brentondwalker/forkulator}}. A user can choose from single-queue, multi-queue, multi-stage, and $(k,l)$ fork-join systems, as well as multi-server systems with thinning. The arrival and service processes can also be specified as constant rate, exponential, Erlang, normal, or Weibull, and optionally regulated through a leaky bucket. The simulator samples jobs at a user-configurable interval and records the sojourn, waiting, and service times.

To ensure that our samples are close to iid, in our experiments we sample every 100th job. We chose that interval based on an empirical analysis of the autocorrelation of sojourn times in trial simulations.  The confidence intervals plotted on all of the simulation results are $68.2\%$ and are computed for the quantile statistics using the method described in~\cite[Sec. 2.2.2]{boudec-performance-evaluation}.  In all cases we ran at least $10^9$ iterations, giving about $10^7$ samples, which is enough to estimate the $1-10^{-6}$ quantile and its confidence interval.

\paragraph*{Spark}Spark~\cite{spark-usenix} is a popular data processing engine that implements the map-reduce model. It is part of the Apache Hadoop ecosystem. Within a Spark program one can execute map and reduce style parallelized operations, and with some APIs one can control the degree of parallelism used.  This allows us to effectively create jobs containing a controllable number of tasks, and set these tasks to execute for controllable lengths of time.  Therefore we can submit a Spark program that is allocated $k$ cores and executes jobs with $k^{\prime}$ tasks, and we can draw the execution times of these tasks from any distribution with non-negative support.  The default task scheduler in Spark puts the tasks in each job in a first-in first-out queue and distributes them to executors as they become available.  This mode of operation corresponds to the non-idling, single-queue fork-join model in Sec.~\ref{sec:forkjoinnonidling}.

We have set up a stand-alone Spark cluster on four 24-core servers.  By running each Spark slave in a Docker container and limiting each slave to a single core, we are able to run at least 15 single-core workers on each node and effectively emulate a 60-node cluster. We configure the host so that each container has its own IP address. This is not a perfect emulation of a 60-node cluster, since the Docker containers on each node share a network stack, but the jobs in our experiments produce very little network traffic and have a trivial reduce stage, so for our purposes this inaccuracy is inconsequential.

Our \texttt{spark-arrivals} program\footnote{Software available at \url{https://github.com/brentondwalker/spark-arrivals}} produces jobs with inter-arrival times drawn from an exponential distribution.  For each job it spawns a new thread which submits a Spark job, parallelized with $k$ tasks (``slices'' in Spark terminology). Running each job in a separate thread is necessary because \texttt{parallelize()} is blocking.  Without multi-threading, each job would only start after the previous job departed, making the system more like a single-queue split-merge system.

The execution time of each task is drawn from either an exponential or Erlang-$k$ distribution.  The tasks generate random numbers and check how long they have been executing. Therefore the time the tasks spend in their execution loops has the desired distribution. However, there is some overhead associated with executing the tasks, changing the service time distribution somewhat. The main components of this are {\it task deserialization time} and {\it scheduler delay}. Task deserialization time is the time needed to distribute the task and associated data to the executor. {\it Scheduler delay} is not as well documented, but appears to be how Spark accounts for any other overhead that is not execution time or deserialization time. These two components combined tended to be about 6ms. In all our experiments the tasks had a mean service time of 1 second, so this overhead was negligible. Since our reduce stage was trivial, the shuffle and result serialization components of the task overhead were always effectively zero. For our Spark experiments we ran $5 \cdot 10^{8}$ iterations and report the $1-10^{-3}$ quantile.
%
%
\section*{Appendix B: Proofs}
\label{sec:appendixb}
\begin{proof}[Proof of Th.~\ref{th:gg1}]
We only show the proof of the sojourn time, as the proof of the waiting time follows similarly.
\paragraph*{G$\mid$G$\mid$1 servers}
We obtain from \eqref{eq:sojourntime} for $\theta > 0$ that
\begin{equation*}
\mathsf{E}\Bigl[e^{\theta T(n)}\Bigr] \le \sum_{\nu=1}^n \mathsf{E}\Bigl[e^{\theta S(\nu,n)}\Bigr] \mathsf{E}\Bigl[e^{-\theta A(\nu,n)}\Bigr] ,
\end{equation*}
where we estimated the maximum by the sum of its arguments and used the statistical independence of arrivals and service. By insertion of the $(\sigma,\rho)$-constraints from Def.~\ref{def:sigmarho} we have
\begin{equation*}
\!\mathsf{E}\Bigl[e^{\theta T(n)}\!\Bigr] \! \le \! e^{\theta (\sigma_A(-\theta) + \sigma_S(\theta) + \rho_S(\theta))} \! \sum_{\nu=1}^n \! e^{-\theta (\rho_A(-\theta) - \rho_S(\theta)) (n-\nu)}.
\end{equation*}
Next, we estimate
\begin{align*}
\sum_{\nu=1}^n e^{-\theta (\rho_A(-\theta)-\rho_S(\theta)) (n-\nu)} & \le \sum_{\nu=0}^{\infty} \bigl(e^{-\theta (\rho_A(-\theta) - \rho_S(\theta))}\bigr)^{\nu} \\
& = \frac{1}{1-e^{-\theta (\rho_A(-\theta) - \rho_S(\theta))}},
\end{align*}
where we used the geometric sum for $\rho_S(\theta) < \rho_A(-\theta)$. By use of Chernoff's bound $\mathsf{P}[X \ge x] \le e^{-\theta x} \mathsf{E}\bigl[e^{\theta X}\bigr]$ we obtain
\begin{equation*}
\mathsf{P}[T(n) \ge \tau] \le \frac{e^{\theta (\sigma_A(-\theta) + \sigma_S(\theta))}}{1-e^{-\theta (\rho_A(-\theta) - \rho_S(\theta))}} e^{\theta \rho_S(\theta)} e^{-\theta\tau}.
\end{equation*}
\paragraph*{GI$\mid$GI$\mid$1 servers}
From \eqref{eq:sojourntime} we have
\begin{equation*}
T(n) = \max_{m \in [1,n]} \{S(n-m+1,n) - A(n-m+1,n) \} .
\end{equation*}
For $\theta > 0$ we can write
\begin{multline*}
\mathsf{P}[T(n) > \tau ] \\ = \mathsf{P}\left[\max_{m \in [1,n]} \left\{e^{\theta (S(n-m+1,n) - A(n-m+1,n))} \right\} > e^{\theta\tau} \right] .
\end{multline*}
Now consider the process
\begin{equation*}
U(m) = e^{\theta (S(n-m+1,n) - A(n-m+1,n))} .
\end{equation*}
Using the representation of $A(m,n) = \sum_{\nu=m}^{n-1} A(\nu,\nu+1)$ and $S(m,n) = \sum_{\nu=m}^n S(\nu)$ by increment processes, we have
\begin{equation*}
U(m+1) = U(m) e^{\theta (S(n-m) - A(n-m,n-m+1))} .
\end{equation*}
The conditional expectation can be computed as
\begin{align*}
&\mathsf{E}[U(m+1) | U(m), U(m-1),\dots,U(1)] \\
=& U(m) \mathsf{E}\Bigl[ e^{\theta S(n-m)}\Bigr] \mathsf{E}\Bigl[ e^{-\theta  A(n-m,n-m+1)}\Bigr] ,
\end{align*}
where we used the independence of the inter-arrival times and the service times. If $\rho_S(\theta) \le \rho_A(-\theta)$, it holds that $\mathsf{E}\bigl[ e^{\theta S(n-m)}\bigr] \mathsf{E}\bigl[ e^{-\theta  A(n-m,n-m+1)}\bigr] \le 1$ and
\begin{equation*}
\mathsf{E}[U(m+1) | U(m), U(m-1),\dots,U(1)] \le U(m),
\end{equation*}
i.e., $U(m)$ is a supermartingale. By application of Doob's inequality for submartingales~\cite[Theorem 3.2, p. 314]{doob:stochasticprocesses} and the formulation for supermartingales~\cite{jiang:onecoin, jiang:noteonsnetcalc} we have for non-negative $U(m)$ for $m \ge 1$ that
\begin{equation}
x \mathsf{P} \left[ \max_{m \in [1,n]} \{ U(m) \} \ge x \right] \le \mathsf{E}[U(1)] .
\label{eq:martingalebound}
\end{equation}
We derive
\begin{equation*}
\mathsf{E}[U(1)] = \mathsf{E}\Bigl[e^{\theta (S(n,n) - A(n,n))}\Bigr] = \mathsf{E}\Bigl[e^{\theta S(1)}\Bigr].
\end{equation*}
Letting $x = e^{\theta \tau}$ we have from \eqref{eq:martingalebound} that
\begin{equation*}
\mathsf{P} \left[ T(n) \ge \tau \right] \le e^{\theta \rho_S(\theta)} e^{-\theta\tau} ,
\end{equation*}
which completes the proof of Th.~\ref{th:gg1}.
\end{proof}
%
%
\begin{proof}[Proof of Th.~\ref{th:multistage}]
First, we derive the MGF of $S^{\text{net}}(m,n)$ as in Lem.~\ref{lem:multistage}. It follows for $\theta > 0$ that
\begin{multline*}
\mathsf{E}\Bigl[e^{\theta S^{\text{net}}(m,n)}\Bigr] \le \sum_{\nu^j : m \le \nu^1 \le \nu^2 \le \dots \le \nu^{h-1} \le n} \mathsf{E}\Bigl[e^{\theta S^{1}(m,\nu^1)}\Bigr] \\ \mathsf{E}\Bigl[e^{\theta S^{2}(\nu^1,\nu^2)}\Bigr] \cdots \mathsf{E}\Bigl[e^{\theta S^{h}(\nu^{h-1},n)}\Bigr] ,
\end{multline*}
where we estimated the maximum by the sum of its arguments and used the statistical independence of the stages. After some variable substitutions we obtain
\begin{multline*}
\mathsf{E}\Bigl[e^{\theta S^{\text{net}}(m,n)}\Bigr] \le \sum_{\nu^j \ge 0: \sum_{j=1}^h \nu^j = n-m} \mathsf{E}\Bigl[e^{\theta S^{1}(m,m+\nu^1)}\Bigr] \\ \mathsf{E}\Bigl[e^{\theta S^{2}(m+\nu^1,m+\nu^1+\nu^2)}\Bigr] \cdots \mathsf{E}\Bigl[e^{\theta S^{h}(m+\sum_{j=1}^{h-1} \nu^j,m+\sum_{j=1}^h \nu^j)}\Bigr] .
\end{multline*}
Given homogeneous stages that are $(\sigma_{S},\rho_{S})$ constrained as specified by Def.~\ref{def:sigmarho} and using~\cite[Prop. 6.2]{ross:probability} to replace the sum by a binomial coefficient, we have for $\theta > 0$ that
\begin{equation*}
\mathsf{E}\Bigl[e^{\theta S^{\text{net}}(m,n)}\Bigr] \le {n-m+h-1 \choose h-1} e^{\theta (h \sigma_{S}(\theta) + \rho_{S}(\theta) (n-m+h))} .
\end{equation*}

Next, we derive for the MGF of the sojourn time from \eqref{eq:sojourntime} for $\theta > 0$ that
\begin{equation*}
\mathsf{E}\Bigl[e^{\theta T(n)}\Bigr] \le \sum_{\nu=1}^n \mathsf{E}\Bigl[e^{\theta S^{\text{net}}(\nu,n)}\Bigr] \mathsf{E}\Bigl[e^{-\theta A(\nu,n)}\Bigr] ,
\end{equation*}
where we estimated the maximum by the sum of its arguments and used the statistical independence of arrivals and service. Considering $(\sigma_A,\rho_A)$ constrained traffic as in Def.~\ref{def:sigmarho}, we have
\begin{multline*}
\mathsf{E}\Bigl[e^{\theta T(n)}\Bigr] \le e^{\theta (\sigma_A(-\theta) + h \sigma_{S}(\theta) + h \rho_{S}(\theta))} \\ \sum_{\nu=1}^n {n-\nu+h-1 \choose h-1} e^{-\theta (\rho_A(-\theta) - \rho_{S}(\theta)) (n-\nu)} .
\end{multline*}
Next, we estimate for $\rho_{S}(\theta) < \rho_A(-\theta)$ that
\begin{align*}
& \sum_{\nu=1}^n {n-\nu+h-1 \choose h-1} e^{-\theta (\rho_A(-\theta)-\rho_{S}(\theta)) (n-\nu)} \\
\le & \sum_{\nu=0}^{\infty} {\nu+h-1 \choose h-1} \Bigl(e^{-\theta (\rho_A(-\theta) - \rho_{S}(\theta))}\Bigr)^{\nu} \\
= & \Bigl(1-e^{-\theta (\rho_A(-\theta) - \rho_{S}(\theta))}\Bigr)^{-h} ,
\end{align*}
where we used that
\begin{multline*}
\sum_{\nu=0}^{\infty} {\nu+h-1 \choose h-1} \\ \Bigl(e^{-\theta (\rho_A(-\theta) - \rho_{S}(\theta))}\Bigr)^{\nu} \Bigl(1-e^{-\theta (\rho_A(-\theta) - \rho_{S}(\theta))}\Bigr)^{h} = 1 ,
\end{multline*}
as the argument of the sum takes the form of the negative binomial probability mass function. By use of Chernoff's bound we obtain
\begin{equation*}
\mathsf{P}[T(n) \ge \tau] \le \frac{e^{\theta (\sigma_A(-\theta) + h \sigma_{S}(\theta))}}{\bigl(1-e^{-\theta (\rho_A(-\theta) - \rho_{S}(\theta))}\bigr)^h} e^{\theta h \rho_{S}(\theta)} e^{-\theta\tau}.
\end{equation*}
Finally, we insert the service parameters of the tasks $\sigma_{S}(\theta) = \sigma_{Q}(\theta) + \ln(k)/\theta$ and $\rho_S(\theta) = \rho_{Q}(\theta)$ from \eqref{eq:forkjoinsigma} and \eqref{eq:forkjoinrho} for each of the fork-join stages to complete the proof.
\end{proof}
%
%
\balance
\bibliographystyle{IEEEtran}
\bibliography{IEEEabrv,ParallelSystems}
%
%
\end{document}